\documentclass[a4paper]{article}       
\usepackage{graphicx,tikz}
\usetikzlibrary{arrows}
\usetikzlibrary{shapes}
\usetikzlibrary{shapes.misc}
\usetikzlibrary{fadings}
\usepackage{amsmath,amssymb,amsthm}
\usepackage{hyperref}
\newcommand\spac{\mathbb{G}}
\newcommand\cspac{0_\spac}
\newcommand\orbi[2]{\mathcal{O}_{#1}\left({#2}\right)}
\newcommand{\Z}{\mathbb{Z}}
\newcommand{\N}{\mathbb{N}}
\newcommand{\boule}[1]{B_{#1}}
\newcommand{\spaclen}[1]{\|#1\|_{\spac}}
\newtheorem{lemma}{Lemma}

\newtheorem{definition}{Definition}
\newtheorem{theorem}{Theorem}

\newtheorem{remark}{Remark}
\newtheorem{question}{Question}
\newtheorem{example}{Example}

\newcommand\confs{Q^\spac}
\newcommand\meas[1]{\mathcal{M}\left({#1}\right)}
\newcommand\trac[2]{\mathcal{T}^{#1}_{#2}} 
\newcommand\freezt[3]{\zeta_{#1}(#2,#3)} 

\newcommand\restr[2]{{
  \left.\kern-\nulldelimiterspace 
  #1 
  \vphantom{|} 
  \right|_{#2} 
  }}

\begin{document}

\title{Cold Dynamics in Cellular Automata: a Tutorial
}


\author{Guillaume Theyssier\thanks{I2M, Université Aix-Marseille, CNRS}}

\maketitle

\begin{abstract}
  This tutorial is about cellular automata that exhibit \emph{cold dynamics}.
  By this we mean zero Kolmogorov-Sinai entropy, stabilization of all orbits, trivial asymptotic dynamics, or similar properties.
  These are essentially transient and irreversible dynamics, but they capture many examples from the literature, ranging from crystal growth to epidemic propagation and symmetric majority vote.
  A collection of properties is presented and discussed: nilpotency and asymptotic, generic or mu- variants, unique ergodicity, convergence, bounded-changeness, freezingness.
  They all correspond to the \emph{cold dynamics} paradigm at various degrees, and we study their links and differences by key examples and results.
  Besides dynamical considerations, we also focus on computational aspects: we show how such \emph{cold cellular automata} can still compute under their dynamical constraints, and what are their computational limitations.
  The purpose of this tutorial is to illustrate how the richness and complexity of the model of cellular automata are preserved under such strong constraints.
  By putting forward some open questions, it is also an invitation to look more closely at this \emph{cold dynamics} territory, which is far from being completely understood.
\end{abstract}

\section{Introduction}

This tutorial is about cellular automata (CA for short), a well-known class of discrete dynamical systems, useful for modeling natural phenomena, and also a model of computation.
Informally, a CA is a lattice of finite-state cells that evolve according to a uniform local rule.
CA have been extensively studied as chaotic dynamical systems \cite{kurkabook}, as models of physical phenomena with positive (Kolmogorov-Sinai) entropy \cite{ChopardDroz,latticegas}, as a model of reversible computing \cite{reversiblecomputing,timesym} or as groups when considering only reversible CA \cite{CAgroupBLR,notits2019}.

None of these kinds of CA are considered here.
On the contrary, we are interested in examples with zero (Kolmogorov-Sinai) entropy, that have a strong convergence property, whose asymptotic dynamics is essentially trivial, and that are strongly irreversible.
These examples are essentially transient dynamically, they erase information from the initial configuration, and display a strong effect of the \emph{arrow of time}.

We will refer to this class as \emph{cold dynamics}, without trying to precisely define it.
It turns out that literature on CA abounds in results and examples that fit into this class.
Actually, examples of CA with such kinds of dynamics were already pointed out by S. Ulam in his pioneering work on crystal (or pattern) growth more than 50 years ago \cite{ulam}.
The computational aspects of some \emph{cold dynamics} CA were also considered more than 40 years ago \cite{vollmar81} with the point of view of language recognition.
Since then, many theoretical works have focused on some particular \emph{cold} properties, the flagship result among them probably being the undecidability of nilpotency\footnote{This notion will be presented in Section~\ref{sec:nil}.}, established 30 years ago \cite{kari92} (and actually much before but not explicitly in \cite{aanderaa_lewis_1974}).

The purpose of this tutorial is to present and organize this literature around key concepts like convergence, and to illustrate how the richness and complexity of the model of cellular automata is preserved despite the strong constraints behind the \emph{cold dynamics} paradigm.

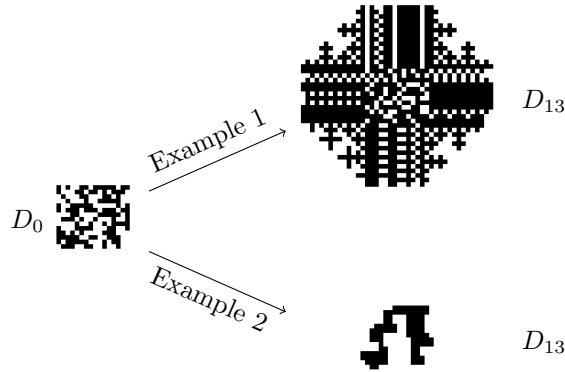
\begin{figure}
  \centering
  \begin{tikzpicture}[scale=.8]
    \draw (0,0) node {\includegraphics[width=3cm]{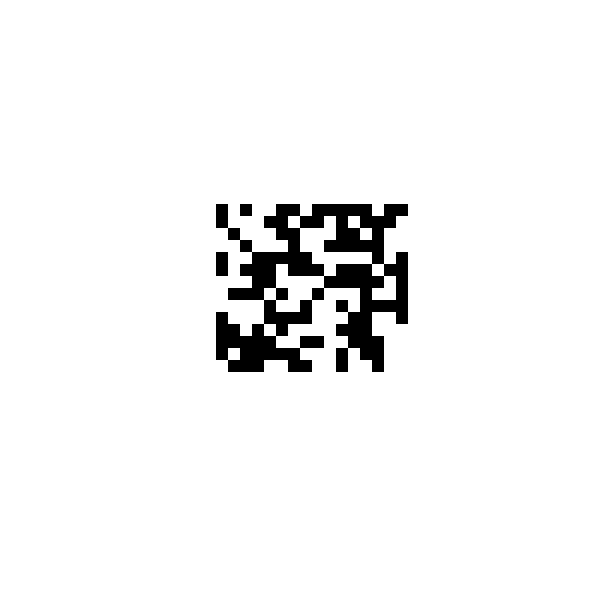}};
    \draw (5,-2) node {\includegraphics[width=3cm]{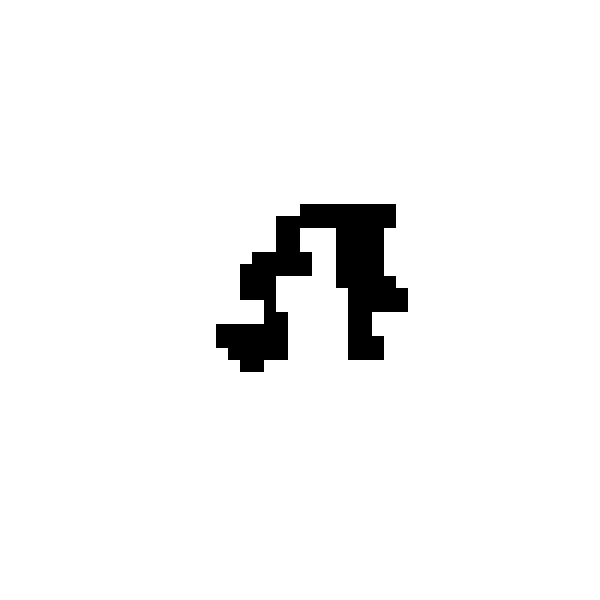}};
    \draw (5,2) node {\includegraphics[width=3cm]{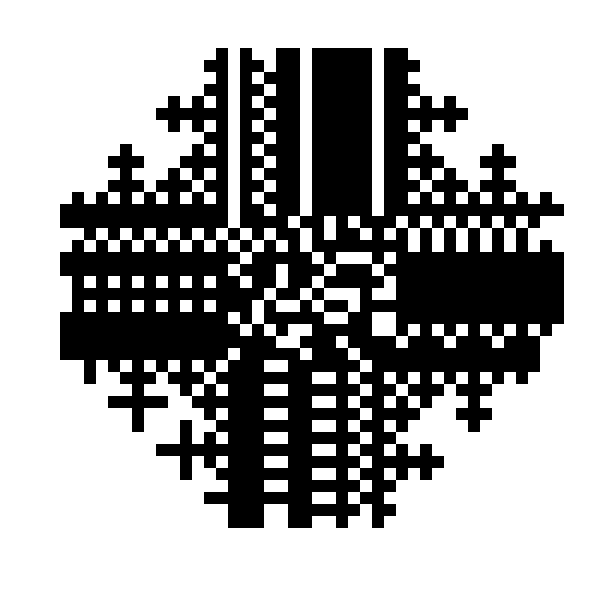}};
    \draw (-1,0) node {$D_0$};
    \draw[->] (1,.5) --node[midway,sloped,above] {Example 1} (3.25,1.5);
    \draw (7.5,2) node {$D_{13}$};
    \draw[->] (1,-.5) --node[midway,sloped,below] {Example 2} (3.25,-1.5);
    \draw (7.5,-2) node {$D_{13}$};
  \end{tikzpicture}
  \caption{Evolution starting from the same initial configuration for the two examples of cold dynamics presented in this section. Positions belonging to sets $D_i$ are represented in black.}
  \label{fig:exintro}
\end{figure}

\paragraph{Concrete examples.} Let us consider the discrete plane $\Z^2$ and say that two elements are neighbors if one is at distance 1 to the north, east, south or west of the other.
Now consider any initial set ${D_0\subseteq\Z^2}$ and define $D_{i+1}$ for any $i\in\N$ as the union of $D_i$ and all elements of $\Z^2$ that have exactly one neighbor in $D_i$.
This actually defines a CA with a cold dynamics.
Indeed, the sets $D_i$ are only increasing and therefore each position of $\Z^2$ can only have one of two destinies: either it will never belong to any $D_i$, or it will appear in some $D_i$ and then stay in all subsequent ones.
This first example was introduced in \cite{ulam} by S. Ulam and it is actually one of the very first CA ever defined.
It was thought of as a toy model of crystal growth.

As our second example, consider again any initial set ${D_0\subseteq\Z^2}$ and define $D_{i+1}$ as set of positions from $\Z^2$ such that, among itself and its $4$ neighbors, a majority of positions is in $D_i$ (\textit{i.e.} at least $3$).
In this second example, the destiny of a position is less clear as the local majority could change several times.
It turns out that if we consider one step every two, this example also has a strong convergence property (we will establish this for a large class of examples in Theorem~\ref{theo:majboundedchange}).
Note that majority dynamics like this one might well be the single most studied class of CA.

An example of evolution of both CA is show in Figure~\ref{fig:exintro}.
Other examples of two states CA  with a \emph{cold dynamics} have been specifically studied \cite{ChLeRe79,Holroyd03,GriMoo96}.
In the context of modeling through CA some propagation phenomena (epidemic, rumor, fire, etc), many models have been proposed following the so-called compartmental paradigm \cite{KMcK27}: the population is divided into categories (like Susceptible, Infected and Recovered for the SIR model) and the model describes the local conditions for an individual to change from a category to another.
When the spatial aspects are considered, this kind of models naturally give rise to multi-state CA and many of them have a \emph{cold dynamics} (see Figure~\ref{fig:forestfire}).

\begin{figure}
  \centering
  \includegraphics[width=.7\textwidth]{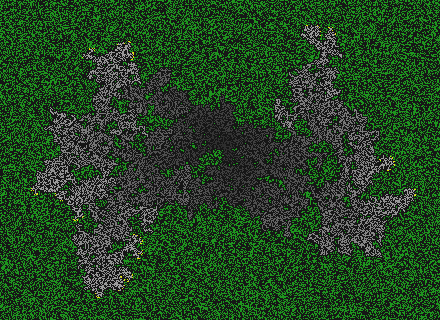}  
  \caption{A configuration reached by a naive CA model of forest fire: fire (yellow) propagates over neighboring trees (in green) and become ashes (gray) in one step. Shades of gray indicates the time since the fire left the position (the darker the sooner). The density of trees has a huge impact on the fire propagation (in this example the initial density of trees is $0.4$).}
  \label{fig:forestfire}
\end{figure}

One could start to analyze these specific examples in depth, but we will rather try to understand what are the properties that make them special, and work on entire classes of examples.

\paragraph{Main questions around cold dynamics.} This tutorial is organized into five parts: we will first precise our formal settings in Section~\ref{sec:definitions} (using the language of orbits, topology and probability measures) and introduce in Section~\ref{sec:convergent} the pivotal notion for cold dynamics: convergence. Then, we will address three general questions as starting points to present many results and examples:
\begin{enumerate}
\item \textit{What kind of behaviors can be obtained if we push the convergence property to an extreme and ask that the system collapse all initial configurations into a single one?
  What if we only specify that the asymptotic dynamics be trivial (a singleton) and how does it depend on the language used (configurations versus probability measures)?}

  We will tackle these questions in Section~\ref{sec:nil} by studying the classical notion of nilpotency and its variants.
\item \textit{How to establish that a given CA is convergent?}
  
  In Section~\ref{sec:proveconv}, we will both present proof techniques and introduce useful sub-classes of convergent CA (bounded-change and freezing CA) that put stronger constraints on the dynamics but already capture many natural examples, like the two presented above.
\item \textit{How does the ability of general CA to make arbitrary computations survive under the constraint of convergence, bounded-change or freezingness?}
  
  In Section~\ref{sec:compuconv}, we will see that even under the strongest constraints, computational universality is preserved, but the effectiveness of this computational power greatly varies with the constraints considered.
\end{enumerate}

\paragraph{A tutorial.} This tutorial is an invitation to discover what we believe are nice examples, concepts and proof technics.
It is not a survey, and we try to find a good compromise between accessibility and generality.
We do include some proofs (or sketch of) as they give much insight in the topic.
However, many results are stated without proof.
This is mostly due to space constraint and not an indication that the result is less important: the reader is warmly invited to consult the corresponding references.

\section{Definitions and Notations}
\label{sec:definitions}

A (discrete) dynamical system is a pair ${(X,F)}$ where $X$ is a compact metric space (seen as a set of possible global states of the system) and ${F:X\to X}$ a continuous map (seen as the action or global evolution rule of the system).
Given an initial point ${x\in X}$, its \emph{orbit} is the sequence $x$, $F(x)$, $F^2(x)$, $F^3(x)$, etc.
It represents the evolution of the system at successive \emph{time steps} when initialized in $x$.

We will consider only two kinds of dynamical systems in this tutorial: deterministic CA acting on \emph{configurations}, and deterministic CA acting on \emph{probability measures}.
A CA is defined by some \emph{uniform spatial structure} of cells, an \emph{alphabet} and a \emph{local evolution rule} that defines its dynamics.
For the reader unfamiliar with these notions, let us start by a concrete example to fix ideas:
\begin{itemize}
\item choose the alphabet to be ${\{0,1\}}$,
\item choose the spatial structure of cells to be a bi-infinite line,
\item configurations are then bi-infinite sequences of $0$s and $1$s,
\item choose the following local rule that defines how the CA transforms a configuration into a new one at each time step: each cell takes as new state the sum modulo $2$ of its own state and the one of its right neighbor. Denote by $F$ this \emph{global evolution rule} on configuration.
\end{itemize}

\newcommand\poi[2]{\fill[red] (#1,#2) circle (1pt);}
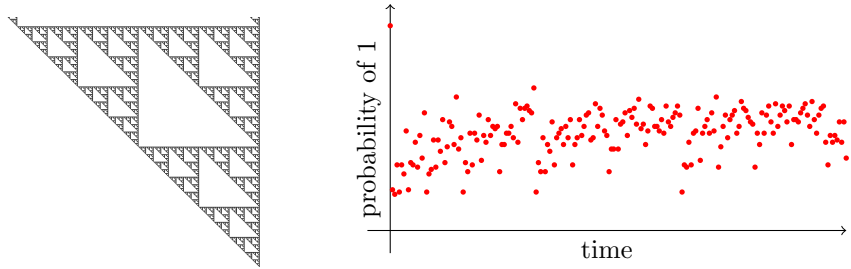
\begin{figure}
  \centering
  \includegraphics[width=.3\linewidth]{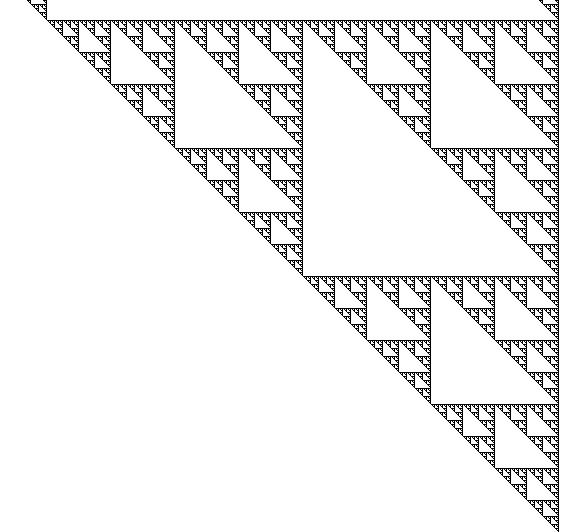}\hskip 1cm
  \begin{minipage}[t]{.6\linewidth}
    \begin{tikzpicture}[x=3cm,y=3cm]
      \draw[->] (-.1,0) -- node[midway,below] {time} +(2.1,0);
      \draw[->] (0,-.1) -- node[midway,sloped,above] {probability of $1$} +(0,1.1);
      \poi{0.000000}{0.905000}
      \poi{0.010000}{0.180000}
      \poi{0.020000}{0.160000}
      \poi{0.030000}{0.290000}
      \poi{0.040000}{0.170000}
      \poi{0.050000}{0.290000}
      \poi{0.060000}{0.250000}
      \poi{0.070000}{0.440000}
      \poi{0.080000}{0.180000}
      \poi{0.090000}{0.300000}
      \poi{0.100000}{0.290000}
      \poi{0.110000}{0.390000}
      \poi{0.120000}{0.280000}
      \poi{0.130000}{0.420000}
      \poi{0.140000}{0.320000}
      \poi{0.150000}{0.520000}
      \poi{0.160000}{0.170000}
      \poi{0.170000}{0.250000}
      \poi{0.180000}{0.270000}
      \poi{0.190000}{0.400000}
      \poi{0.200000}{0.280000}
      \poi{0.210000}{0.400000}
      \poi{0.220000}{0.350000}
      \poi{0.230000}{0.490000}
      \poi{0.240000}{0.300000}
      \poi{0.250000}{0.370000}
      \poi{0.260000}{0.480000}
      \poi{0.270000}{0.460000}
      \poi{0.280000}{0.380000}
      \poi{0.290000}{0.590000}
      \poi{0.300000}{0.360000}
      \poi{0.310000}{0.410000}
      \poi{0.320000}{0.170000}
      \poi{0.330000}{0.300000}
      \poi{0.340000}{0.260000}
      \poi{0.350000}{0.430000}
      \poi{0.360000}{0.290000}
      \poi{0.370000}{0.430000}
      \poi{0.380000}{0.400000}
      \poi{0.390000}{0.510000}
      \poi{0.400000}{0.310000}
      \poi{0.410000}{0.420000}
      \poi{0.420000}{0.390000}
      \poi{0.430000}{0.490000}
      \poi{0.440000}{0.420000}
      \poi{0.450000}{0.460000}
      \poi{0.460000}{0.450000}
      \poi{0.470000}{0.460000}
      \poi{0.480000}{0.260000}
      \poi{0.490000}{0.380000}
      \poi{0.500000}{0.380000}
      \poi{0.510000}{0.430000}
      \poi{0.520000}{0.430000}
      \poi{0.530000}{0.460000}
      \poi{0.540000}{0.410000}
      \poi{0.550000}{0.560000}
      \poi{0.560000}{0.390000}
      \poi{0.570000}{0.540000}
      \poi{0.580000}{0.490000}
      \poi{0.590000}{0.540000}
      \poi{0.600000}{0.550000}
      \poi{0.610000}{0.530000}
      \poi{0.620000}{0.520000}
      \poi{0.630000}{0.630000}
      \poi{0.640000}{0.170000}
      \poi{0.650000}{0.300000}
      \poi{0.660000}{0.260000}
      \poi{0.670000}{0.410000}
      \poi{0.680000}{0.260000}
      \poi{0.690000}{0.380000}
      \poi{0.700000}{0.350000}
      \poi{0.710000}{0.480000}
      \poi{0.720000}{0.290000}
      \poi{0.730000}{0.410000}
      \poi{0.740000}{0.430000}
      \poi{0.750000}{0.450000}
      \poi{0.760000}{0.380000}
      \poi{0.770000}{0.460000}
      \poi{0.780000}{0.410000}
      \poi{0.790000}{0.490000}
      \poi{0.800000}{0.290000}
      \poi{0.810000}{0.380000}
      \poi{0.820000}{0.410000}
      \poi{0.830000}{0.490000}
      \poi{0.840000}{0.420000}
      \poi{0.850000}{0.460000}
      \poi{0.860000}{0.510000}
      \poi{0.870000}{0.520000}
      \poi{0.880000}{0.400000}
      \poi{0.890000}{0.410000}
      \poi{0.900000}{0.560000}
      \poi{0.910000}{0.460000}
      \poi{0.920000}{0.540000}
      \poi{0.930000}{0.510000}
      \poi{0.940000}{0.440000}
      \poi{0.950000}{0.420000}
      \poi{0.960000}{0.260000}
      \poi{0.970000}{0.360000}
      \poi{0.980000}{0.360000}
      \poi{0.990000}{0.490000}
      \poi{1.000000}{0.360000}
      \poi{1.010000}{0.470000}
      \poi{1.020000}{0.480000}
      \poi{1.030000}{0.520000}
      \poi{1.040000}{0.420000}
      \poi{1.050000}{0.530000}
      \poi{1.060000}{0.500000}
      \poi{1.070000}{0.540000}
      \poi{1.080000}{0.470000}
      \poi{1.090000}{0.580000}
      \poi{1.100000}{0.450000}
      \poi{1.110000}{0.470000}
      \poi{1.120000}{0.380000}
      \poi{1.130000}{0.440000}
      \poi{1.140000}{0.550000}
      \poi{1.150000}{0.540000}
      \poi{1.160000}{0.460000}
      \poi{1.170000}{0.460000}
      \poi{1.180000}{0.440000}
      \poi{1.190000}{0.490000}
      \poi{1.200000}{0.430000}
      \poi{1.210000}{0.550000}
      \poi{1.220000}{0.480000}
      \poi{1.230000}{0.460000}
      \poi{1.240000}{0.500000}
      \poi{1.250000}{0.520000}
      \poi{1.260000}{0.550000}
      \poi{1.270000}{0.510000}
      \poi{1.280000}{0.170000}
      \poi{1.290000}{0.290000}
      \poi{1.300000}{0.280000}
      \poi{1.310000}{0.390000}
      \poi{1.320000}{0.290000}
      \poi{1.330000}{0.430000}
      \poi{1.340000}{0.440000}
      \poi{1.350000}{0.540000}
      \poi{1.360000}{0.310000}
      \poi{1.370000}{0.430000}
      \poi{1.380000}{0.460000}
      \poi{1.390000}{0.490000}
      \poi{1.400000}{0.430000}
      \poi{1.410000}{0.520000}
      \poi{1.420000}{0.440000}
      \poi{1.430000}{0.590000}
      \poi{1.440000}{0.280000}
      \poi{1.450000}{0.370000}
      \poi{1.460000}{0.460000}
      \poi{1.470000}{0.510000}
      \poi{1.480000}{0.410000}
      \poi{1.490000}{0.490000}
      \poi{1.500000}{0.510000}
      \poi{1.510000}{0.530000}
      \poi{1.520000}{0.430000}
      \poi{1.530000}{0.460000}
      \poi{1.540000}{0.570000}
      \poi{1.550000}{0.540000}
      \poi{1.560000}{0.530000}
      \poi{1.570000}{0.500000}
      \poi{1.580000}{0.480000}
      \poi{1.590000}{0.470000}
      \poi{1.600000}{0.280000}
      \poi{1.610000}{0.430000}
      \poi{1.620000}{0.380000}
      \poi{1.630000}{0.520000}
      \poi{1.640000}{0.430000}
      \poi{1.650000}{0.510000}
      \poi{1.660000}{0.550000}
      \poi{1.670000}{0.480000}
      \poi{1.680000}{0.430000}
      \poi{1.690000}{0.560000}
      \poi{1.700000}{0.460000}
      \poi{1.710000}{0.510000}
      \poi{1.720000}{0.550000}
      \poi{1.730000}{0.530000}
      \poi{1.740000}{0.480000}
      \poi{1.750000}{0.550000}
      \poi{1.760000}{0.380000}
      \poi{1.770000}{0.490000}
      \poi{1.780000}{0.490000}
      \poi{1.790000}{0.480000}
      \poi{1.800000}{0.520000}
      \poi{1.810000}{0.540000}
      \poi{1.820000}{0.560000}
      \poi{1.830000}{0.510000}
      \poi{1.840000}{0.460000}
      \poi{1.850000}{0.570000}
      \poi{1.860000}{0.490000}
      \poi{1.870000}{0.430000}
      \poi{1.880000}{0.510000}
      \poi{1.890000}{0.540000}
      \poi{1.900000}{0.550000}
      \poi{1.910000}{0.390000}
      \poi{1.920000}{0.290000}
      \poi{1.930000}{0.420000}
      \poi{1.940000}{0.390000}
      \poi{1.950000}{0.420000}
      \poi{1.960000}{0.400000}
      \poi{1.970000}{0.480000}
      \poi{1.980000}{0.390000}
      \poi{1.990000}{0.480000}
      \poi{2.000000}{0.320000}
    \end{tikzpicture}
  \end{minipage}
  \caption{The sum modulo 2 CA seen as two dynamical systems. On the left: example of an orbit of a configuration represented as a space-time diagram (time goes from bottom to top, each horizontal line represents a configuration, white represents state 0 and black state 1). On the right: example of an orbit of a measure (only a partial information on measures is represented).}
  \label{fig:xorbits}
\end{figure}

Starting from a given initial configuration and applying iteratively the global transformation $F$, we get a sequence of configurations, \textit{i.e.} an \emph{orbit} (see Figure~\ref{fig:xorbits}).
For instance, if $c$ is the configuration where each cell is in state $0$, then ${F(c)=c}$ and we say that $c$ is a \emph{fixed point} of $F$.

We can also see our example CA as a deterministic dynamical system acting on \emph{probability measures}.
To simplify exposition, let's think of a measure $\mu$ as a process that produces (random) configurations.
Then one can consider a new measure $F(\mu)$ which is the process that first produces a random configuration by $\mu$ and then applies $F$ to get a new random configuration.
From there we can of course apply iteratively $F$ and get a sequence of measures: $F$ can thus also be seen as a dynamical system acting on measures.
For instance, consider $\mu_{p}$ the random process that chooses the state of each cell independently to be $0$ with probability $p$ and $1$ with probability $1-p$ (such measures are called Bernoulli measures).
It can be shown that ${F(\mu_{1/2})=\mu_{1/2}}$, said differently $\mu_{1/2}$ is an \emph{invariant measure} for $F$.
However ${F(\mu_{1/4})\neq\mu_{1/4}}$ because when the states of a given cell and its right neighbor are chosen independently to be $0$ with probability $1/4$ and $1$ with probability $3/4$, then after one step of $F$ the considered cell has a probability ${\frac{1}{4}\times\frac{3}{4} + \frac{3}{4}\times\frac{1}{4} = \frac{3}{8}}$ to be in state $1$.
A measure can convey much more information than the probability for an individual cell to be in some state: for instance the reader can verify that $F(\mu_{1/4})$ is not of the form $\mu_p$ for any $p$ (the probability of states in two adjacent cells are no longer independent).\\

In the remaining of this section we present standard definitions about CA and establish a precise formal framework as well as some notations.
Most of the results presented in this tutorial do not require in-depth knowledge of the notions defined below.
The reader expecting a more detailed presentation is invited to consult \cite{cagroups,Pivato09,lindmarcus,kurkabook,walters1981}.

\paragraph{Spatial structure of cells.} A natural and fairly general settings for the spatial cell structure is to consider a finitely generated group $\spac$ (without mentioning it each time, all groups $\spac$ below will be finitely generated).
We will use additive notation for groups and denote by ${\cspac}$ the identity of ${\spac}$.
For ${z\in\spac}$, we denote by ${\spaclen{z}}$ the length $n$ of the smallest sequence of generators ${(g_i)_{1\leq i\leq n}}$ such that ${g_1+\cdots +g_n=z}$, and we let ${\spaclen{\cspac}=0}$.
We denote by ${\boule{n}}$ the set of elements ${z\in\spac}$ such that ${\spaclen{z}\leq n}$.
Most of the time we will consider ${\spac=\Z^d}$ and discuss about cases ${d=1}$ or ${d\geq2}$ (called 1D CA, 2D CA, etc).
We will only scratch the surface of the deep and largely open question of how CA theory depends on $\spac$.
We will sometimes use the notion of amenability which is an emblematic example of this dependence through the so-called \emph{Garden of Eden} theorem (see \cite{cagroups}).
A finitely additive probability measure on $\spac$ is a map $\mu$ from subsets of $\spac$ to ${[0,1]}$ such that ${\mu(\spac)=1}$ and ${\mu(A\cup B)=\mu(A)+\mu(B)}$ whenever ${A\cap B=\emptyset}$.
A finitely additive probability measure $\mu$ is said left-invariant if ${\mu(A) = \mu(g + A)}$ for any ${g\in\spac}$ and any ${A\subseteq\spac}$.
Then we say that $\spac$ is \emph{amenable} if admits a left-invariant finitely additive probability measure.
Groups ${\Z^d}$ are amenable and we will sometimes consider free groups with $2$ or more generators as an example of non-amenable group.

\paragraph{Space of configurations and its topology.} Given a finite alphabet $Q$ and a group $\spac$, a \emph{configuration} is a map ${c:\spac\to Q}$.
In order to remove some parentheses in expressions, we will often denote ${c(z)}$ by ${c_z}$.
For any ${q\in Q}$, we denote by ${\overline{q}}$ the uniform configuration ${c:z\mapsto q}$.
The set of configurations ${\confs}$ can be endowed with the prodiscrete topology (product of the discrete topology on $Q$) to form a compact space \cite{kurkabook}.
Given a finite set ${D\subseteq\spac}$ and a partial configuration ${u: D\to Q}$, called a \emph{pattern}, we denote by ${[u]}$ the \emph{cylinder set} of configurations that coincide with $u$ on domain $D$: 
\[[u] = \{c\in\confs : \restr{c}{D} = u\}.\]
By a slight abuse of notation, we denote by $[q]$ for $q\in Q$ the cylinder set ${[u]}$ of domain ${\{\cspac\}}$ where ${u = \cspac\mapsto q}$.
Cylinder sets are clopen sets and form a base of the prodiscrete topology.
This topology is also metrizable as follows. For any pair of configurations ${x,y\in\confs}$ define their distance by 
\[d(x,y) =
  \begin{cases}
    0 &\text{ if $x=y$,}\\
    2^{-n} & \text{ where $n=\max\{i : \restr{x}{\boule{i}} = \restr{y}{\boule{i}}\}$}.
  \end{cases}
\]
We will often consider sequences of configurations ${(x_n)_n}$ that converges to some limit, denoted ${x_n\to_n x}$ or ${\displaystyle\lim_{n\to\infty} x_n = x}$ in the sequel: it will always be in the sense of this prodiscrete topology.
${x_n\to_n x}$ is equivalent to the property that for each ${g\in\spac}$ the sequence of states ${x_n(g)}$ converges to ${x(g)}$, \textit{i.e.} is ultimately constant equal to ${x(g)}$.
$\spac$ naturally acts on $\confs$ by \emph{translation} as follows: for any ${z\in\spac}$ we define the shift map ${\sigma_z:\confs\to\confs}$ by 
\[\sigma_z(c)_{z'} = c_{z+z'}.\]
Shift maps are homeomorphisms. For ${c\in\confs}$, we denote by ${\orbi{\spac}{c}}$ its orbit under translations, that is: ${\orbi{\spac}{c}=\{\sigma_z(c) : z\in\spac\}}$.
A natural kind of subset of configurations are closed translation invariant subsets of $\confs$. Such a subset $X$ is called a \emph{subshift} and it is characterized by its \emph{forbidden language}, \textit{i.e.} the set of patterns it avoids: ${\{u : [u]\cap X=\emptyset\}}$ (see \cite{lindmarcus}).

\paragraph{Evolution rule.} A CA on ${\confs}$ is defined by a local evolution rule ${\lambda : Q^N\to Q}$ where $N$ is a finite subset of $\spac$ called \emph{neighborhood}, and whose associated global evolution rule ${F:\confs\to\confs}$ is defined as follows: 
\[F(c)_z = \lambda\bigl(\restr{\sigma_z(c)}{N}\bigr)\]
for all ${z\in\spac}$.
We say that $F$ has \emph{radius} $r$ if ${N\subseteq\boule{r}}$ (note that the same map $F$ could be defined by different local maps $\lambda$ with different neighborhoods $N$, but finding the minimal such $N$ and hence the minimal radius is usually not important in the sequel).
Any CA map $F$ is continuous and commutes with translations: ${\sigma_z\circ F = F\circ\sigma_z}$.
Then ${\bigl(\confs,F)}$ is a dynamical system.
Conversely, by Curtis-Hedlund-Lyndon theorem \cite{hedlund,cagroups}, any dynamical system ${\bigl(\confs,\phi)}$ where $\phi$ commutes with translations is actually a CA.
A CA is said bijective (resp. surjective, resp. injective) if its corresponding global map on configuration is.
A fundamental consequence of Curtis-Hedlund-Lyndon theorem is that if $F$ is the global map of a bijective CA, then the inverse map $F^{-1}$ is also the global map of a CA.
For this reason, bijective CA are often referred to as \emph{reversible CA}.

Given a CA ${F:\confs\to\confs}$, we are naturally interested in orbits, \textit{i.e.} sequences of the form ${c,F(c),\ldots, F^t(c),\ldots}$ for some configuration $c$.
A partial information on orbits called \emph{canonical trace} will play an important role: it is the sequence ${\trac{F}{c}:\N\to Q}$ defined as ${c_{\cspac}, F(c)_{\cspac},\ldots,F^t(c)_{\cspac},\ldots}$ for a given $c\in\confs$.
The map ${c\mapsto\trac{F}{c}}$ is a factor map from ${\confs\to Q^\N}$. It is generally neither surjective, nor injective.

\paragraph{Space of probability measures and action of a CA on it.} The set of Borel probability measures on $\confs$ will be denoted by ${\meas{\confs}}$.
A measure ${\mu\in\meas{\confs}}$ is a countably additive function from Borel sets to ${[0,1]}$ such that ${\mu(\confs)=1}$.
The \emph{support} of a measure is the smallest closed set $X$ such that ${\mu(X)=1}$.
We say that a measure has \emph{full-support} if its support is ${\confs}$.
Concretely, a measure is characterized by its value on cylinders (by the Carathéodory-Fréchet extension theorem since cylinders form a semi-ring that generates the $\sigma$-algebra of Borel sets).
It is also convenient to define a measure this way.
For instance, a natural and important class of measures are the \emph{Bernoulli measures} which are product measures defined by coefficients ${\beta_q}$ for each ${q\in Q}$ such that ${\sum_{q\in Q}\beta_q=1}$ as follows: for any cylinder $[u]$ with ${u:D\to Q}$,
\[\mu([u]) = \prod_{z\in D}\beta_{u_z}.\]
The uniform Bernoulli measure is the one where all coefficients $\beta_q$ are equal to ${1/\#Q}$ where the notation $\#X$ represents the cardinality of the set $X$.
A Bernoulli measure has full-support if and only if ${\beta_q>0}$ for all ${q\in Q}$.
A fundamental property of Bernoulli measures is that they are \emph{translation-ergodic}, meaning that for any Borel set $X$ which is invariant under translation, ${\mu(X)}$ is either $1$ or $0$ \cite{walters1981}.
The somewhat opposite example of measures are dirac measures. The dirac measure associated to a configuration ${c\in\confs}$, denoted $\delta_c$, is defined by 
\[\delta_c(A) =
  \begin{cases}
    1 &\text{ if $c\in A$},\\
    0 &\text{ if $c\not\in A$}
  \end{cases}
\]
for any Borel set $A$.

The set ${\meas{\confs}}$ is naturally endowed with the weak-* topology which is the coarsest topology such that the map ${\mu\mapsto \mu([u])}$ is continuous for each cylinder set $[u]$. This topology is compact and metrizable by the distance 
\[d(\mu,\nu) = \sum_{n\in\N}2^{-n}\max \bigl\{\bigl|\mu([u]) - \nu([u])\bigr| : D\subseteq\boule{n}, u:D\to Q\bigr\}.\]
In this topology, a sequence of measures ${(\mu_n)_n}$ converges to $\mu$ if and only if, for each cylinder set ${[u]}$, ${\mu_n([u])}$ converges to ${\mu([u])}$.

Any CA $F:\confs\to\confs$ acts naturally on the set of measures ${\meas{\confs}}$ as follows: ${F(\mu) = \mu\circ F^{-1}}$ or, said differently, ${F(\mu)}$ is the measure such that for each Borel set one has
\[F(\mu)(X) = \mu\bigl(F^{-1}(X)\bigr).\]
This way ${\bigl(\meas{\confs},F\bigr)}$ is also a dynamical system.
Note that for any ${c\in\confs}$ one has ${F(\delta_c) = \delta_{F(c)}}$ because ${F(c)\in A\iff c\in F^{-1}(A)}$.
Hence the dynamical system ${\bigl(\confs,F\bigr)}$ actually embeds into the dynamical system ${\bigl(\meas{\confs}, F\bigr)}$.

\paragraph{Decision problems and undecidability.} To conclude this definition section, let us briefly clarify the settings behind the decision problems that we will consider.
Our problems deal with CA and/or configurations as input.
A CA is always given by its finite description (alphabet and local evolution rule).
A configuration in input will be given as a map ${c:\spac\to Q}$, which is actually a map from ${A^*}$ to $Q$ where $A$ is a fixed set of generators of $\spac$.
Depending on the context (computable, finite, periodic configuration) there are various ways to represent this map as an input.
In the general case, it will be a Turing machine that computes the map.

Our undecidability results will sometimes refer to the arithmetical hierarchy, we refer to \cite{Rogers} for an in depth presentation.
Intuitively, it is a hierarchy of levels of uncomputability.
Contrary to the polynomial hierarchy in complexity theory, this one is known to be strict.
We will only use the first levels of the hierarchy, for which we can give canonical complete problems about Turing machines to give some intuition:
\begin{itemize}
\item the halting problem is ${\Sigma_1^0}$-complete,
\item the set of Turing machine that halts on all inputs is ${\Pi_2^0}$-complete, its complement is ${\Sigma_2^0}$-complete,
\item the set of Turing machine such that the set of inputs on which it  doesn't halt is infinite is ${\Pi_3^0}$-complete.
\end{itemize}

\section{Convergent CA and their basic properties}
\label{sec:convergent}

The central definition considered in this tutorial is convergence.

\begin{definition}
  A dynamical system ${(X,T)}$ is \emph{convergent} if for any ${x\in X}$ the orbit of $x$ under $T$ converges towards some limit denoted ${T^\omega(x)}$: 
  \[\forall x\in X: \lim_{n\to\infty} T^n(x) = T^\omega(x).\]
\end{definition}

The convergence of a CA ${(\confs,F)}$ is equivalent to the property that for any ${c\in\confs}$, ${\trac{F}{c}}$ is eventually constant. When it is the case, we denote by ${\freezt{F}{c}{z}}$ the \emph{freezing time} of cell $z$ in the orbit of $c$, \textit{i.e.} the first time after which the cell $z$ becomes constant in this orbit: 
\[\freezt{F}{c}{z} = \min \bigl\{t_0\geq 0 : \forall t\geq t_0, F^{t+1}(c)_z = F^{t}(c)_z\bigr\}.\]

We will see various examples of convergent CA in the next sections, but let's start by one of the simplest to fix ideas.

\newcommand\stateZ[2]{\draw[fill=black] (#1,#2) rectangle +(1,1);}
\newcommand\stateO[2]{}
\newcommand\poif[2]{\fill[red] (#1,#2) circle (3pt);}
\begin{figure}
  \centering
  \begin{tikzpicture}[x=.5cm,y=.5cm]
    \draw[->,dotted] (-.5,-.5) -- node[midway,below] {space} +(11,0);
    \draw[->,dotted] (-.5,-.5) -- node[midway,sloped,above] {time} +(0,6);
    \stateZ{0}{0}\stateZ{0}{1}\stateZ{0}{2}\stateZ{0}{3}\stateO{0}{4}\stateZ{1}{0}\stateZ{1}{1}\stateZ{1}{2}\stateO{1}{3}\stateO{1}{4}\stateZ{2}{0}\stateZ{2}{1}\stateO{2}{2}\stateO{2}{3}\stateO{2}{4}\stateZ{3}{0}\stateO{3}{1}\stateO{3}{2}\stateO{3}{3}\stateO{3}{4}\stateO{4}{0}\stateO{4}{1}\stateO{4}{2}\stateO{4}{3}\stateO{4}{4}\stateZ{5}{0}\stateO{5}{1}\stateO{5}{2}\stateO{5}{3}\stateO{5}{4}\stateO{6}{0}\stateO{6}{1}\stateO{6}{2}\stateO{6}{3}\stateO{6}{4}\stateZ{7}{0}\stateZ{7}{1}\stateZ{7}{2}\stateZ{7}{3}\stateZ{7}{4}\stateZ{8}{0}\stateZ{8}{1}\stateZ{8}{2}\stateZ{8}{3}\stateZ{8}{4}\stateZ{9}{0}\stateZ{9}{1}\stateZ{9}{2}\stateZ{9}{3}\stateZ{9}{4}
    \poif{0.5}{4.5}
    \poif{1.5}{3.5}
    \poif{2.5}{2.5}
    \poif{3.5}{1.5}
    \poif{4.5}{0.5}
    \poif{5.5}{1.5}
    \poif{6.5}{0.5}
    \poif{7.5}{0.5}
    \poif{8.5}{0.5}
    \poif{9.5}{0.5}
  \end{tikzpicture}
  \caption{\label{fig:maxfreeze}Freezing time line in an orbit of Example~\ref{ex:max} supposing that all the right part of the configuration not shown here is in state $1$: in black state $1$, in white state $0$ and the red dots represent the freezing time of each cell as positions in space-time.}  
\end{figure}
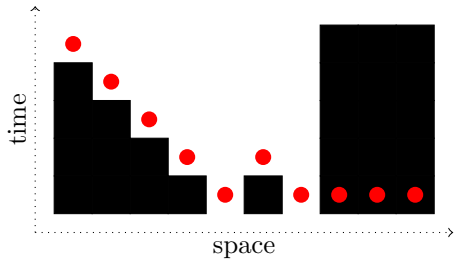

\begin{example}\label{ex:max}
  Let ${Q=\{0,1\}}$ and define the CA $F$ on $\confs$ by 
  \[F(x)_z = \min_{a\in A}x_{z+a}\]
  where $A$ is a fixed finite neighborhood containing the identity of $\spac$.
  Clearly, $F$ can only turn $1$s into $0$s in any configuration, so it is a convergent CA.
  It turns out that this simple dynamical systems actually give much information on the group $\spac$ and its generating sets (see \cite{epperlein2020iterated}).
  If $\spac=\Z$ and ${A=\{0,1\}}$ then the freezing time ${\freezt{F}{c}{z}}$ of cell $z$ starting from configuration $c$ is either $0$ if ${c_z=0}$ or if there is no occurrence of state $0$ to the right of $z$ in $c$, or it is the distance to the closest such occurrence otherwise (see Figure~\ref{fig:maxfreeze}).
  
\end{example}

A first observation is that convergence property on configurations extends to the action on measures.

\begin{lemma}
  \label{lem:convergencemeasure}
  A CA ${(\confs,F)}$ is convergent if and only if ${(\meas{\confs},F)}$ is convergent.
\end{lemma}
\begin{proof}
  Suppose first that ${(\meas{\confs},F)}$ is convergent and consider any ${c\in\confs}$.
  Since ${F^t(\delta_c)= \delta_{F^t(c)}}$ and since ${\bigl(F^t(\delta_c)\bigr)_t}$ converges to some measure ${F^\omega(\delta_c)}$ by hypothesis, we have that for any ${q\in Q}$, ${\delta_{F^t(c)}([q])}$ must converge to ${F^\omega(\delta_c)([q])}$.
  But $\delta_{F^t(c)}([q])$ is either $1$ if ${\trac{F}{c}(t)=q}$ or $0$ otherwise.
  This shows that ${\trac{F}{c}}$ is eventually constant.

  Suppose now for the other direction that ${(\confs,F)}$ is convergent.
  Consider any measure $\mu$ and any pattern ${u\in Q^D}$ with ${D\subseteq \boule{n}}$ for some ${n\in\N}$.
  We want to show that ${\bigl(F^t(\mu)([u])\bigr)_t}$ converges.
  Consider the set $X_t$ of initial configurations whose orbit is frozen for all cells in ${\boule{n}}$ after at most $t$ steps: 
  \[X_t = \bigl\{c\in\confs : \forall z\in\boule{n}, \freezt{F}{c}{z}\leq t\bigr\}.\]
  ${(X_t)_t}$ is monotone increasing and ${\bigcup_tX_t=\confs}$ so the measure of these sets converges to $1$.
  ${(F^{-t}([u])\cap X_t)}$ is also monotone increasing (if ${c\in F^{-t}([u])\cap X_t}$ then ${F^{t'}(c)\in [u]}$ for any ${t'\geq t}$), so the measure of these sets converges (from below) to some $\alpha$.
  Then, for any ${\epsilon>0}$, there is some $t_0\in\N$ such that ${\mu(X_{t_0})\geq 1-\epsilon}$ and ${\alpha-\epsilon\leq\mu(F^{-t_0}([u])\cap X_{t_0})\leq\alpha.}$
  For any ${t\geq t_0}$, we can write ${F^t(\mu)([u])}$ as: 
  \[F^t(\mu)([u]) = \mu\bigl(F^{-t}([u])\cap X_{t_0}\bigr) + \mu\bigl(F^{-t}([u])\setminus X_{t_0}\bigr).\]
  But it actually holds that ${F^{-t}([u])\cap X_{t_0}= F^{-t_0}([u])\cap X_{t_0}}$, so we have actually shown ${|F^t(\mu)([u])-\alpha|\leq\epsilon}$, which concludes convergence of ${(\meas{\confs},F)}$ because $\mu$ and $u$ were chosen arbitrarily.
\end{proof}

A second observation on convergent CA is that they cannot be too chaotic.
Precisely, a dynamical system ${(X,T)}$ is said sensitive to initial conditions if there is $\epsilon$ such that for all $\delta$ and all ${x\in X}$ there is some ${y\in X}$ which is $\delta$-close to $x$ but whose orbit will be $\epsilon$-far from that of $x$ at some future step (see \cite{kurkabook} for more details on this notion in the context of CA).

\newcommand{\trace}[1]{T_{#1}}
In the following lemma, we use the notation ${\trace{A}(c)}$ for some finite set ${A\subseteq\spac}$ to denote the trace of configuration $c$ on $A$, i.e. the map ${t\mapsto \restr{F^t(c)}{A}}$.

\begin{lemma}\label{lem:convnotsens}
  No convergent CA is sensitive to initial conditions.
\end{lemma}
\begin{proof}
  Suppose by contradiction that $F$ on $\confs$ is convergent and sensitive to initial conditions: there is ${N\in\N}$ such that for any ${c\in \confs}$ and any ${p\in\N}$ there exists ${c'\in\confs}$ such that ${\restr{c}{\boule{p}}=\restr{c'}{\boule{p}}}$ and ${\restr{F^t(c)}{\boule{N}}\neq\restr{F^t(c')}{\boule{N}}}$ for some ${t\in\N}$. Consider any configuration ${c_0\in\confs}$ and ${p_1\geq N}$. By sensitivity there is ${c'\in\confs}$ and ${t_1\in\N}$ such that either ${\trace{\boule{N}}(c_0)}$ or ${\trace{\boule{N}}(c')}$ is non-constant on time interval ${[0,t_1]}$. Denote by $c_1$ the one among ${c_0}$ and $c'$ that corresponds to the non-constant trace. Let ${p_2 = p_1 + N + r(t_1+1)}$ where $r$ is the radius of $F$. Applying sensitivity again on $c_1$ and $p_2$ we know there exist $c'$ and $t_2$ such that:
  \begin{itemize}
  \item ${c_1}$ and $c'$ are identical on ${\boule{p_2}}$;
  \item therefore, by choice of $p_2$, ${\trace{\boule{N}}(c_1)}$ and ${\trace{\boule{N}}(c')}$ coincide on interval ${[0,t_1]}$;
  \item ${\trace{\boule{N}}(c_1)}$ and ${\trace{\boule{N}}(c')}$ differ at time ${t_2}$.
  \end{itemize}
  So one of $c_1$ or $c'$, denoted $c_2$, is such that ${\trace{\boule{N}}(c_2)}$ is not constant on interval ${[t_1,t_2]}$. Going on with the same reasoning we construct a converging sequence ${(c_n)_{n\in\N}}$ of configurations such that ${\trace{\boule{N}}(c_n)}$ is not constant on each interval ${[t_i,t_{i+1}]}$ for ${0\leq i<n}$. Taking ${\displaystyle c=\lim_{n\to\infty} c_n}$ we get a trace ${\trace{\boule{N}}(c)}$ which is not eventually constant contradicting convergence.
\end{proof}

Our last observation on convergent CA is that they are always irreversible (meaning non-bijective, see Section~\ref{sec:definitions}) and even non-surjective (except the identity), at least if $\spac$ is amenable.
The proof below essentially relies on Poincaré recurrence theorem.

\begin{lemma}
  Let $\spac$ be an amenable group. If a convergent CA on $\confs$ is surjective, then it is the identity map.
\end{lemma}
\begin{proof}
  Let $F$ be surjective and convergent and let $\mu$ denote the uniform product measure on $\confs$. By \cite{Bartholdi2010}, $\mu$ is preserved under $F$: ${\mu(X)=\mu(F^{-1}(X))}$ for any Borel set $X$ (see \cite{capobianco} for an overview of properties of surjective CA on amenable groups). If we suppose that $F$ is not the identity map, then there is a word ${u\in Q^{\boule{n}}}$ such that for all ${x\in[u]}$ it holds ${F(x)_0\neq x_0}$. We claim that there is a configuration $x$ such that ${F^t(x)\in[u]}$ for infinitely many $t$. From this claim we deduce that $F$ is not convergent because the orbit of $x$ is not convergent, and the lemma follows. To prove the claim, let us denote ${R_t = \bigcup_{t'\geq t}F^{-t'}([u])}$ and ${R=\bigcap_tR_t}$.  By definition $R$ is the set of configurations whose orbit visits $[u]$ infinitely many times and we want to show that ${R\neq\emptyset}$. We actually show that ${\mu(R)>0}$. Since ${R_{t+1}=F^{-1}(R_t)}$ for all ${t\geq 0}$ we have ${\mu(R_t)=\mu(R_0)}$. Moreover ${[u]\subseteq R_0}$ so we deduce that ${\mu([u]\setminus R_t)\leq \mu(R_0\setminus R_t)= 0}$ since ${R_t\subseteq R_0}$ and ${\mu(R_t)=\mu(R_0)}$. Finally, by Boole's inequality we have ${\mu([u]\setminus R)=0}$ so ${\mu(R)>0}$ since ${\mu([u])>0}$.
\end{proof}

\section{Nilpotency and its Variants}
\label{sec:nil}

This section explores examples where the asymptotic dynamics is reduced to a singleton.
We will consider various properties using either the language of orbits of configurations or measures, or the language of traces.

\subsection{Extreme cases of convergence}

Let us start by the strongest forms of nilpotency, which are strengthened forms of convergence.

\begin{definition}
  Let $F$ be be a CA on $\confs$.
  \begin{itemize}
  \item $F$ is said \emph{nilpotent} if there is $x\in\confs$ and
    ${t_0\in\N}$ such that for all ${y\in\confs}$ it holds ${F^{t_0}(y)=x}$.
  \item $F$ is said \emph{asymptotically nilpotent} if there
    is ${x\in\confs}$ such that for all ${y\in\confs}$ it holds
    ${F^t(y)\to_t x}$.
\end{itemize}
\end{definition}

\renewcommand\stateZ[2]{\draw[fill=white] (#1,#2) rectangle +(1,1);}
\renewcommand\stateO[2]{\draw[fill=blue] (#1,#2) rectangle +(1,1);}
\newcommand\stateZO[2]{\draw[fill=red] (#1,#2) rectangle +(1,1);}
\newcommand\etz{\tikz[baseline,scale=.3]{\stateZ{0}{0}}}
\newcommand\etu{\tikz[baseline,scale=.3]{\stateO{0}{0}}}
\newcommand\etd{\tikz[baseline,scale=.3]{\stateZO{0}{0}}}

\begin{example}
  A local evolution map which is a constant map yields a nilpotent CA, but let us give a non-trivial example. Consider the CA ${F:Q^\Z\to Q^\Z}$ of radius $1$ where ${Q=\bigl\{\etz,\etu,\etd\bigr\}}$ and defined by all the transitions appearing in the following space-time diagram (times goes from bottom to top) and such that any transition not appearing in it produces the state $\etz$:
  \begin{center}
    \begin{tikzpicture}[scale=.2]
      \stateZ{0}{0}\stateZ{0}{1}\stateZ{0}{2}\stateZ{0}{3}\stateZ{0}{4}\stateZ{0}{5}\stateZ{0}{6}\stateZ{0}{7}\stateZ{0}{8}\stateZ{0}{9}\stateZ{0}{10}\stateZ{0}{11}\stateZ{0}{12}\stateZ{0}{13}\stateZ{0}{14}\stateZ{0}{15}\stateZ{0}{16}\stateZ{0}{17}\stateZ{0}{18}\stateZ{0}{19}\stateZ{1}{0}\stateZ{1}{1}\stateZ{1}{2}\stateZ{1}{3}\stateZ{1}{4}\stateZ{1}{5}\stateZ{1}{6}\stateZ{1}{7}\stateZ{1}{8}\stateZ{1}{9}\stateZ{1}{10}\stateZ{1}{11}\stateZ{1}{12}\stateZ{1}{13}\stateZ{1}{14}\stateZ{1}{15}\stateZ{1}{16}\stateZ{1}{17}\stateZ{1}{18}\stateZ{1}{19}\stateZ{2}{0}\stateZ{2}{1}\stateZ{2}{2}\stateZ{2}{3}\stateZ{2}{4}\stateZ{2}{5}\stateZ{2}{6}\stateZ{2}{7}\stateZ{2}{8}\stateZ{2}{9}\stateZ{2}{10}\stateZ{2}{11}\stateZ{2}{12}\stateZ{2}{13}\stateZ{2}{14}\stateZ{2}{15}\stateZ{2}{16}\stateZ{2}{17}\stateZ{2}{18}\stateZ{2}{19}\stateZ{3}{0}\stateZ{3}{1}\stateZ{3}{2}\stateZ{3}{3}\stateZ{3}{4}\stateZ{3}{5}\stateZ{3}{6}\stateZ{3}{7}\stateZ{3}{8}\stateZ{3}{9}\stateZ{3}{10}\stateZ{3}{11}\stateZ{3}{12}\stateZ{3}{13}\stateZ{3}{14}\stateZ{3}{15}\stateZ{3}{16}\stateZ{3}{17}\stateZ{3}{18}\stateZ{3}{19}\stateZO{4}{0}\stateZ{4}{1}\stateZ{4}{2}\stateZ{4}{3}\stateZ{4}{4}\stateZ{4}{5}\stateZ{4}{6}\stateZ{4}{7}\stateZ{4}{8}\stateZ{4}{9}\stateZ{4}{10}\stateZ{4}{11}\stateZ{4}{12}\stateZ{4}{13}\stateZ{4}{14}\stateZ{4}{15}\stateZ{4}{16}\stateZ{4}{17}\stateZ{4}{18}\stateZ{4}{19}\stateZO{5}{0}\stateO{5}{1}\stateO{5}{2}\stateO{5}{3}\stateZ{5}{4}\stateZ{5}{5}\stateZ{5}{6}\stateZ{5}{7}\stateZ{5}{8}\stateZ{5}{9}\stateZ{5}{10}\stateZ{5}{11}\stateZ{5}{12}\stateZ{5}{13}\stateZ{5}{14}\stateZ{5}{15}\stateZ{5}{16}\stateZ{5}{17}\stateZ{5}{18}\stateZ{5}{19}\stateO{6}{0}\stateZO{6}{1}\stateZO{6}{2}\stateZ{6}{3}\stateZ{6}{4}\stateZ{6}{5}\stateZ{6}{6}\stateZ{6}{7}\stateZ{6}{8}\stateZ{6}{9}\stateZ{6}{10}\stateZ{6}{11}\stateZ{6}{12}\stateZ{6}{13}\stateZ{6}{14}\stateZ{6}{15}\stateZ{6}{16}\stateZ{6}{17}\stateZ{6}{18}\stateZ{6}{19}\stateZ{7}{0}\stateZO{7}{1}\stateZ{7}{2}\stateO{7}{3}\stateO{7}{4}\stateZ{7}{5}\stateZ{7}{6}\stateZ{7}{7}\stateZ{7}{8}\stateZ{7}{9}\stateZ{7}{10}\stateZ{7}{11}\stateZ{7}{12}\stateZ{7}{13}\stateZ{7}{14}\stateZ{7}{15}\stateZ{7}{16}\stateZ{7}{17}\stateZ{7}{18}\stateZ{7}{19}\stateZO{8}{0}\stateZ{8}{1}\stateO{8}{2}\stateZO{8}{3}\stateZ{8}{4}\stateZO{8}{5}\stateO{8}{6}\stateZO{8}{7}\stateZ{8}{8}\stateZ{8}{9}\stateZ{8}{10}\stateZ{8}{11}\stateZ{8}{12}\stateZ{8}{13}\stateZ{8}{14}\stateZ{8}{15}\stateZ{8}{16}\stateZ{8}{17}\stateZ{8}{18}\stateZ{8}{19}\stateZO{9}{0}\stateO{9}{1}\stateO{9}{2}\stateO{9}{3}\stateZO{9}{4}\stateZ{9}{5}\stateO{9}{6}\stateO{9}{7}\stateZO{9}{8}\stateZ{9}{9}\stateZ{9}{10}\stateZ{9}{11}\stateZ{9}{12}\stateZ{9}{13}\stateZ{9}{14}\stateZ{9}{15}\stateZ{9}{16}\stateZ{9}{17}\stateZ{9}{18}\stateZ{9}{19}\stateO{10}{0}\stateZO{10}{1}\stateZO{10}{2}\stateZ{10}{3}\stateZO{10}{4}\stateO{10}{5}\stateZO{10}{6}\stateZ{10}{7}\stateZO{10}{8}\stateO{10}{9}\stateZO{10}{10}\stateO{10}{11}\stateO{10}{12}\stateZ{10}{13}\stateZ{10}{14}\stateZ{10}{15}\stateZ{10}{16}\stateZ{10}{17}\stateZ{10}{18}\stateZ{10}{19}\stateZ{11}{0}\stateZO{11}{1}\stateZ{11}{2}\stateZO{11}{3}\stateO{11}{4}\stateO{11}{5}\stateZ{11}{6}\stateZO{11}{7}\stateO{11}{8}\stateO{11}{9}\stateZ{11}{10}\stateZO{11}{11}\stateZ{11}{12}\stateZO{11}{13}\stateO{11}{14}\stateO{11}{15}\stateZ{11}{16}\stateZ{11}{17}\stateZ{11}{18}\stateZ{11}{19}\stateZO{12}{0}\stateZ{12}{1}\stateZO{12}{2}\stateZ{12}{3}\stateZO{12}{4}\stateZ{12}{5}\stateZO{12}{6}\stateZ{12}{7}\stateZO{12}{8}\stateZ{12}{9}\stateZO{12}{10}\stateZ{12}{11}\stateZO{12}{12}\stateZ{12}{13}\stateZO{12}{14}\stateZ{12}{15}\stateZO{12}{16}\stateZ{12}{17}\stateZ{12}{18}\stateZ{12}{19}\stateO{13}{0}\stateZO{13}{1}\stateO{13}{2}\stateZO{13}{3}\stateO{13}{4}\stateZO{13}{5}\stateO{13}{6}\stateZO{13}{7}\stateO{13}{8}\stateZO{13}{9}\stateO{13}{10}\stateZO{13}{11}\stateO{13}{12}\stateZO{13}{13}\stateO{13}{14}\stateZO{13}{15}\stateO{13}{16}\stateZO{13}{17}\stateO{13}{18}\stateZ{13}{19}\stateZ{14}{0}\stateZ{14}{1}\stateZ{14}{2}\stateZ{14}{3}\stateZ{14}{4}\stateZ{14}{5}\stateZ{14}{6}\stateZ{14}{7}\stateZ{14}{8}\stateZ{14}{9}\stateZ{14}{10}\stateZ{14}{11}\stateZ{14}{12}\stateZ{14}{13}\stateZ{14}{14}\stateZ{14}{15}\stateZ{14}{16}\stateZ{14}{17}\stateZ{14}{18}\stateZ{14}{19}\stateZ{15}{0}\stateZ{15}{1}\stateZ{15}{2}\stateZ{15}{3}\stateZ{15}{4}\stateZ{15}{5}\stateZ{15}{6}\stateZ{15}{7}\stateZ{15}{8}\stateZ{15}{9}\stateZ{15}{10}\stateZ{15}{11}\stateZ{15}{12}\stateZ{15}{13}\stateZ{15}{14}\stateZ{15}{15}\stateZ{15}{16}\stateZ{15}{17}\stateZ{15}{18}\stateZ{15}{19}\stateZ{16}{0}\stateZ{16}{1}\stateZ{16}{2}\stateZ{16}{3}\stateZ{16}{4}\stateZ{16}{5}\stateZ{16}{6}\stateZ{16}{7}\stateZ{16}{8}\stateZ{16}{9}\stateZ{16}{10}\stateZ{16}{11}\stateZ{16}{12}\stateZ{16}{13}\stateZ{16}{14}\stateZ{16}{15}\stateZ{16}{16}\stateZ{16}{17}\stateZ{16}{18}\stateZ{16}{19}\stateZ{17}{0}\stateZ{17}{1}\stateZ{17}{2}\stateZ{17}{3}\stateZ{17}{4}\stateZ{17}{5}\stateZ{17}{6}\stateZ{17}{7}\stateZ{17}{8}\stateZ{17}{9}\stateZ{17}{10}\stateZ{17}{11}\stateZ{17}{12}\stateZ{17}{13}\stateZ{17}{14}\stateZ{17}{15}\stateZ{17}{16}\stateZ{17}{17}\stateZ{17}{18}\stateZ{17}{19}\stateZ{18}{0}\stateZ{18}{1}\stateZ{18}{2}\stateZ{18}{3}\stateZ{18}{4}\stateZ{18}{5}\stateZ{18}{6}\stateZ{18}{7}\stateZ{18}{8}\stateZ{18}{9}\stateZ{18}{10}\stateZ{18}{11}\stateZ{18}{12}\stateZ{18}{13}\stateZ{18}{14}\stateZ{18}{15}\stateZ{18}{16}\stateZ{18}{17}\stateZ{18}{18}\stateZ{18}{19}\stateZ{19}{0}\stateZ{19}{1}\stateZ{19}{2}\stateZ{19}{3}\stateZ{19}{4}\stateZ{19}{5}\stateZ{19}{6}\stateZ{19}{7}\stateZ{19}{8}\stateZ{19}{9}\stateZ{19}{10}\stateZ{19}{11}\stateZ{19}{12}\stateZ{19}{13}\stateZ{19}{14}\stateZ{19}{15}\stateZ{19}{16}\stateZ{19}{17}\stateZ{19}{18}\stateZ{19}{19}
    \end{tikzpicture}
  \end{center}
  We claim that ${F^{19}(c)_0 = \etz}$ for any configuration $c$ so that $F$ is nilpotent (we verified this claim by a computer program).
  However, the diagram above clearly shows that ${F^{18}}$ is not a constant map.
  We challenge the reader to find a nilpotent CA with same alphabet and radius but such that ${F^{19}}$ is not a constant map. 
\end{example}

The attracting configuration $x$ in the definitions above must be uniform, because it must attract in particular uniform configurations. Thus, these definitions can be expressed as properties of traces:
\begin{itemize}
\item $F$ is nilpotent if and only if there is $t_0$ and $q\in Q$
  such that ${\trac{F}{x}(t)=q}$ for all ${t\geq t_0}$ and all
  ${x\in\confs}$.
\item $F$ is asymptotically nilpotent if and only if there is $q\in Q$
  such that all
  ${x\in\confs}$ there is $t_x\in\N$ such that ${\trac{F}{x}(t)=q}$ for all ${t\geq t_x}$.
\end{itemize}

Note that the state $q$ in the above definitions must be \emph{quiescent}, \textit{i.e.} such that the uniform configuration $\overline{q}$ is a fixed point of $F$.
The property of nilpotency is also well-known as a property of the \emph{limit set}.
The limit set of a CA $F$ is the set ${\Omega_F = \bigcap_t F^t(\confs)}$ (see \cite{Culik89,hu87}).
Nilpotency is equivalent to ${\Omega_F}$ being a singleton.

Asymptotic nilpotency can be reformulated through the action of the CA on measures.

\begin{lemma}
  \label{lem:asynilmeasure}
  A CA $F$ on $\confs$ is asymptotically nilpotent if and only there is $q\in Q$ such that for any ${\mu\in\meas{\confs}}$ it holds ${F^t(\mu)\to_t\delta_{\overline{q}}}$.
\end{lemma}
\begin{proof}
  The property on measures implies asymptotic nilpotency by the same argument as the first part of the proof of Lemma~\ref{lem:convergencemeasure} with the additional fact that ${F^\omega(\delta_c)=\delta_{\overline{q}}}$ which implies that ${F^\omega(c)=\overline{q}}$.

  Conversely, if we suppose that $F$ is asymptotically nilpotent, we can use the same reasoning as in part two of the proof of Lemma~\ref{lem:convergencemeasure}. The additional fact in this case is that ${F^\omega(c)=\overline{q}}$ for all ${c\in\confs}$ so that ${F^{-t}([u])\cap X_t=\emptyset}$ for all pattern $u$ except the uniform one: ${z\in\boule{n}\mapsto q}$. We deduce that ${F^t(\mu)}$ converges towards ${\delta_{\overline{q}}}$.
\end{proof}

It turns out that, under some conditions on $\spac$, nilpotency and asymptotic nilpotency are equivalent.
One might believe that this can be easily proved by a standard compacity argument, but it is not the case.

\begin{theorem}[\cite{GuillonR08,villeasnil12}]\label{theo:asnilimpnil}
  If ${\spac=\Z^d}$ then nilpotency is equivalent to asymptotic nilpotency.
\end{theorem}
\begin{proof}[sketch]
  Let $q$ be the quiescent state involved in the asymptotic nilpotency property.
  First, we remark that for any ball $\boule{n}$ there is a uniform bound $T_n$ such that in the orbit of any configuration $c$ there is some time step ${t\leq T_n}$ for which the cells inside $\boule{n}$ are all in state $q$: ${\forall z\in\boule{n}, F^t(c)_z = q}$.
  Indeed, otherwise we would obtain by compacity a configuration whose orbit would not converge to $\overline{q}$.
  Call \emph{finite} any configuration which is everywhere $q$ except on a finite region, and call \emph{mortal} any configuration $c$ such that there is $t$ with ${F^t(c)=\overline{q}}$.
  
  Let us first establish the theorem for ${\spac=\Z}$ as in \cite{GuillonR08}.
  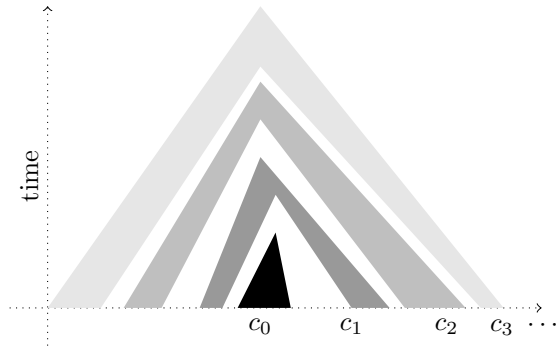
\begin{figure}
    \centering
    \begin{tikzpicture}
      \draw[dotted,->] (-1,0)--(6,0);
      \draw[dotted,->] (-.5,-.5)-- node[midway,sloped,above] {time} (-.5,4);
      \fill[black] (2,0)--(2.5,1)--(2.7,0)--cycle;
      \fill[gray!80!white] (1.5,0)--(1.8,0)--(2.5,1.5)--(3.5,0)--(4,0)--(2.3,2)--cycle;
      \fill[gray!50!white] (.5,0)--(1,0)--(2.3,2.5)--(4.2,0)--(5,0)--(2.3,3)--cycle;
      \fill[gray!20!white] (-.5,0)--(.2,0)--(2.3,3.2)--(5.2,0)--(5.5,0)--(2.3,4)--cycle;
      \draw (2.3,-.25) node {$c_0$};
      \draw (3.5,-.25) node {$c_1$};
      \draw (4.75,-.25) node {$c_2$};
      \draw (5.2,-.25) node[right] {$c_3\ \cdots$};
    \end{tikzpicture}
    \caption{The infinite nested space-time diagram contradicting asymptotic nilpotency in Theorem~\ref{theo:asnilimpnil}. The white part represents state $q$, the grayed parts represent states other than $q$, and each nuance of gray corresponds to a particular configuration $c_i$.}
    \label{fig:nested}
  \end{figure}
  There are two key arguments:
  \begin{itemize}
  \item first, we show nilpotency on finite mortal configurations, \textit{i.e.} there is a uniform $T$ such that for all finite mortal $c$ it holds ${F^T(c)=\overline{q}}$.
    This is proved using the remark above about times $T_n$: if there is no uniform time bound on the mortality of finite mortal configuration, then for arbitrary large $\boule{n}$ and arbitrary large $t$ we can find a finite mortal configuration $c$ such that $c$ is uniformly $q$ on $\boule{n}$ but ${F^t(c)_0\neq q}$.
    From there, we can produce an infinite \textit{nested space-time diagram} as in Figure~\ref{fig:nested} by taking a well-chosen sequence $(c_i)$ of such configurations and considering their superposition $c$: it is sufficient that the non-$q$ portions of space-time diagrams are disjoint enough among the $c_i$ (taking into account the radius of the CA).
    The orbit of this configuration $c$ gives a contradiction with the convergence hypothesis.
    
  \item second, it can be shown that if the CA is not nilpotent, then we can extract finite mortal configurations that die arbitrarily late: indeed, by enclosing a finite configuration alive at time $t$ with blocking words (whose existence is granted by Lemma~\ref{lem:convnotsens}, see \cite{kurkabook}), we obtain a mortal configuration still alive at time $t-b$ where $b$ is a constant depending only on the choice of blocking words. This contradicts the previous item, so we deduce that the CA is nilpotent.
  \end{itemize}
  To extend the result to higher dimensions, it is enough to show that all finite configurations are mortal (this together with asymptotic nilpotency implies nilpotency, by the nested construction argument).
  To simplify we restrict to dimension $2$ and consider any asymptotically nilpotent CA $F$.
  The idea is to reduce to the one-dimensional case by remarking that if a finite configuration $c$ is periodized vertically to form a configuration $c'$, then $c'$ is mortal because the action of $F$ on $c'$ can be seen as the action of a one-dimensional asymptotically nilpotent CA on a one-dimensional configuration.
  The key step of the proof is then to show that no finite configuration can spread too much in the horizontal direction.
  This step is proved by contradiction using again a nested construction of an infinite space-time diagram, but this time a vertical periodization is applied at each step to guarantee mortality.
  Of course, this argument on the bounded horizontal spreading of finite configurations is also valid vertically, so we deduce that the orbit of any finite configuration remains inside a finite support forever.
  This is enough to conclude mortality since $F$ is supposed asymptotically nilpotent.
\end{proof}

The above result has been greatly generalized by \cite{nilrigidity} in at least two ways: the class of dynamical systems considered (for instance, one can consider CA defined over a subshift rather than the whole space $\confs$) and the class of groups $\spac$ considered.
To obtain these results, a notion of tiered dynamical systems was developed which goes way beyond the scope of this tutorial.
The authors also coined the term \emph{nil-rigidity} and used it with various classes of dynamical systems.
Let us stick to the simplest setting and say that a group $\spac$ is nilrigid if any asymptotically nilpotent CA on $\confs$ is nilpotent.
Despite the many results obtained about nilrigidity in \cite{nilrigidity} they left some interesting open questions (see section 11), some of which where already asked in \cite{sanville}.

\begin{question}[\cite{nilrigidity,sanville}]
  Is there any finitely generated group $\spac$ which is not nilrigid? In particular, is the free group with $2$ or more generators nilrigid?
\end{question}

\subsection{Variations with partial convergence and measures}

Let us now consider variations of these strong notions of nilpotency.
We can first consider a similar property of orbits but restricted to a subset of initial configurations.
Let us begin by two properties related to dense subsets of configurations (at least when ${\spac=\Z^2}$).

A configuration is $q$-\emph{finite} for some state $q$ if it is everywhere $q$ except on a finite region of $\spac$.
A configuration $x\in\confs$ is \emph{totally periodic} if its orbit by translation ${\orbi{\spac}{x}}$ is finite.
When ${\spac=\Z^d}$ a totally periodic configuration is a configuration with $d$ non-colinear periods.

\begin{definition}\label{def:nilper}
  \begin{itemize}
  \item $F$ is \emph{nilpotent on periodic configurations} if there is
    a configuration $c$ such that for any totally periodic configuration $x$ there is $t$ with ${F^t(x)=c}$.
  \item $F$ is \emph{nilpotent on finite configurations} if there is
    a state $q$ such that for any $q$-finite configuration $x$ there is $t$ with ${F^t(x)=\overline{q}}$.
\end{itemize}
\end{definition}

Note again that the attracting configuration $c$ must be uniform in the first definition.
Nilpotency on finite configurations is also sometimes called \emph{eroder property} (see Section~\ref{sec:related}).
An asymptotically nilpotent CA must be nilpotent on periodic configurations and nilpotent on finite configurations, but the converse is far from being true as we will see below.

As a second (seemingly unrelated) variation around nilpotency, we can change the point of view and switch from configurations to measures.
A measure ${\mu\in\meas{\confs}}$ is invariant for a CA on ${\confs}$ if ${F(\mu)=\mu}$. The following definition is classical in dynamical system theory \cite{walters1981} and can be defined as a rigidity property of invariant measures.

\begin{definition}\label{def:uniqergo}
  $F$ is \emph{uniquely ergodic} if it possesses a unique invariant measure.
\end{definition}

The invariant measures of $F$ are exactly the limit points of Cesaro mean sequences ${(u_n)_{n\in\N}}$ of the form (\cite[Theorem 6.9]{walters1981})
\[u_n= \frac{1}{n}\sum_{t=0}^{n-1}F^t\mu\]
for ${\mu\in\meas{\confs}}$.

Orbit-wise unique ergodicity is equivalent to the convergence of all such Cesaro mean sequences to the same limit. 
Therefore we could also name this property \emph{Cesaro-nilpotency} to stress the dynamical side of it but it is best known as \emph{unique ergodicity}.
Note that by Lemma~\ref{lem:asynilmeasure}, any asymptotically nilpotent CA is also uniquely ergodic.

Nilpotency over periodic configurations and unique ergodicity are both implied by asymptotic nilpotency. The next lemma shows that the three properties actually form a hierarchy when ${\spac=\Z^d}$.

\begin{lemma}
  A uniquely ergodic CA with ${\spac=\Z^d}$ is nilpotent on periodic configurations.
  \label{lem:uniqergonilperio}
\end{lemma}
\begin{proof}
  Consider any uniquely ergodic CA $F$ on $\confs$.
  Suppose first that there are two disjoint temporal cycles of totally periodic configurations, \textit{i.e.} two sequences of totally periodic configurations ${(c_0,\ldots,c_{k-1})}$ and ${(d_0,\ldots,d_{l-1})}$ with ${F(c_i)= c_{i+1\bmod k}}$ and ${F(d_j)=d_{j+1\bmod l}}$ for ${0\leq i<k}$ and ${0\leq j<l}$.
  Then we can define two measures ${\mu_c = \frac{1}{k}\sum_{i=0}^{k-1}\delta_{c_i}}$ and ${\mu_d = \frac{1}{l}\sum_{j=0}^{l-1}\delta_{d_j}}$ that are both invariant under $F$ because $(c_i)$ and $(d_j)$ are cyclic orbits.
  Moreover since the two cycles  are disjoint there must exist a cylinder $[u]$ such that $c_i\in[u]$ for some $i$ but ${d_j\not\in[u]}$ whatever $j$.
  Thus ${\mu_c([u])>0}$ and ${\mu_d([u])=0}$ which proves that $\mu_c$ and $\mu_d$ are two distinct invariant measures: a contradiction.

  We have shown that $F$ has the following synchronization property:
  \begin{center}
    \textit{there is a cycle of totally periodic configurations ${(c_0,\ldots,c_{k-1})}$ such that the orbit of any totally periodic configuration eventually enters this cycle $(c_i)$.}
  \end{center}
  We claim that ${k=1}$ which concludes the proof since it means that $F$ is nilpotent on periodic configurations.

  We are actually going to show the stronger claim that no CA $F$ (not necessarily uniquely ergodic) with ${\spac=\Z^d}$ can have the synchronization property above with ${k\geq 2}$.
  The argument we give is due to G. Richard (see \cite{Fates19} for more background on the synchronization problem).
  First note that if this synchronization property holds for some CA $F$ with ${\spac=\Z^d}$ and ${d>1}$ then it must also hold for some CA with ${\spac=\Z}$.
  Indeed, considering $F$ on configurations which are constant in ${d-1}$ directions and periodic in the last one, we actually define a one-dimensional CA which has the synchronization property as $F$ does.
  We can thus suppose without loss of generality that $d=1$.
  Clearly, each configuration $c_i$ in the attracting cycle must be uniform because any uniform configuration has an orbit eventually entering this cycle.
  However, if the synchronization property holds with ${k\geq 2}$ then $F$ has no quiescent state.
  Now consider ${f:Q^m\to Q}$ the local map of $F$ and let $w\in Q^\N$ be a semi-infinite word verifying 
  \[w_{n+m} = f(w_n,\ldots,w_{n+m-1})\]
  for all ${n\in\N}$.
  There must exist $n_1$ and $n_2$ with ${n_1+m<n_2}$ such that \[w_{[n_1,n_1+m-1]}=w_{[n_2,n_2+m-1]}.\] Consider the word $w=w_{[n_1,n_2-1]}$ of length at least ${m+1}$.
  By construction and since $F$ has no quiescent state $w$ contains at least two distinct letters from $Q$.
  Moreover, if $c$ is a periodic $\Z$-configuration of period $w$ then ${F(c) = \sigma_p(c)}$ for some $p$.
  Therefore the orbit of $c$ is a cycle of non-uniform configurations, which contradicts the synchronization property. 
\end{proof}

Interestingly, we can also capture unique ergodicity as a property of traces of orbits on configurations: it is the property that there is some attracting state whose frequency goes to $1$ when considering larger and larger prefixes of any trace.
To be more precise, for any ${q\in Q}$ and any word $w\in Q^\N$ we denote by ${d_q(w)}$ the (inferior) asymptotic density of occurrences of $q$ in $w$: 
\[d_q(w)=\liminf_n \frac{\#\{0\leq i<n : w_i=q\}}{n}.\]

We then get this equivalent formulation of unique ergodicity \cite[Proposition 3.2]{TORMA2015415}.

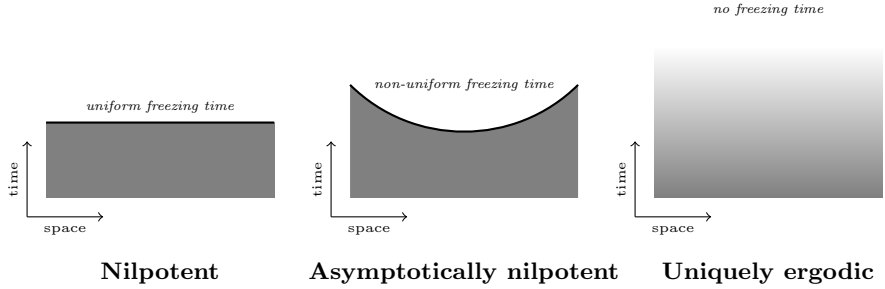
\begin{figure}
  \centering
  \begin{tikzpicture}
    \fill[fill=gray] (0,0) rectangle (3,1);
    \draw[thick] (0,1)--node[midway,above] {\tiny\it uniform freezing time} (3,1);
    \fill[fill=gray] (4,0)--(7,0)--(7,1.5) to[out=225,in=315] (4,1.5) -- cycle;
    \draw[thick] (7,1.5) to[out=225,in=315] (4,1.5);
    \draw (5.5,1.5) node {\tiny\it non-uniform freezing time};
    \fill[fill=gray,path fading=north] (8,0) rectangle (11,2);
    \draw (9.5,2.5) node {\tiny\it no freezing time};
    \draw[->] (-.25,-.25)-- node[midway,below] {\tiny space} +(1,0);
    \draw[->] (-.25,-.25)-- node[midway,sloped,above] {\tiny time} +(0,1);
    \draw (1.5,-1) node {\small\textbf{Nilpotent}};
    \begin{scope}[shift={(4,0)}]
      \draw[->] (-.25,-.25)-- node[midway,below] {\tiny space} +(1,0);
      \draw[->] (-.25,-.25)-- node[midway,sloped,above] {\tiny time} +(0,1);
      \draw (1.5,-1) node {\small\textbf{Asymptotically nilpotent}};
    \end{scope}
    \begin{scope}[shift={(8,0)}]
      \draw[->] (-.25,-.25)-- node[midway,below] {\tiny space} +(1,0);
      \draw[->] (-.25,-.25)-- node[midway,sloped,above] {\tiny time} +(0,1);
      \draw (1.5,-1) node {\small\textbf{Uniquely ergodic}};
    \end{scope}
  \end{tikzpicture}
  \caption{Properties of any space-time diagram of a nilpotent, asymptotically nilpotent and uniquely ergodic CA respectively. White color represents the quiescent state, gray color any state and various shades of gray corresponds to various density of the quiescent state (the lighter, the higher the density). Time goes from bottom to top.}
  \label{fig:nilorbits}
\end{figure}

\begin{lemma}
  \label{lem:traceuniqergo}
  A CA $F$ on $\confs$ with ${\spac=\Z^d}$ is uniquely ergodic if and only if there is ${q\in Q}$ such that for any ${x\in\confs}$ it holds ${d_q(\trac{F}{x}) = 1}$.
\end{lemma}
\begin{proof}
  Suppose first that $F$ is uniquely ergodic, then by Lemma~\ref{lem:uniqergonilperio} it is nilpotent on periodic configurations and has a (unique) quiescent state $q$ and its unique invariant measure is $\delta_{\overline{q}}$.
  Consider any ${x\in\confs}$ and any ${\epsilon>0}$.
  If there were infinitely many ${n\in\N}$ such that ${\frac{\#\{0\leq i<n : \trac{F}{x}(i)\neq q\}}{n}>\epsilon}$, then we could extract a limit point $\mu$ from the sequence ${\bigl(\frac{1}{n}\sum_{t=0}^{n-1}F^t(\delta_x)\bigr)_n}$ which would be an invariant measure \cite[Theorem 6.9]{walters1981} such that ${\mu([q])<1-\epsilon}$ (recall that ${F^t(\delta_x)=\delta_{F^t(x)}}$) hence different from $\delta_{\overline{q}}$: a contradiction.
  We deduce that ${d_q(\trac{F}{x})=1}$.

  Suppose now that ${d_q(\trac{F}{x})=1}$ for all ${x\in\confs}$ for some ${q\in Q}$. Then $q$ must be a quiescent state and thus the measure $\delta_{\overline{q}}$ is invariant and also ergodic.
  If it were not the unique invariant measure then there would be another invariant ergodic measure $\mu$ (see \cite[Theorem 6.10]{walters1981}), which would necessarily verify ${\mu([g\mapsto q])<1}$ for some ${g\in\spac}$ in order to differ from ${\delta_{\overline{q}}}$ (otherwise the Choquet's decomposition theorem on invariant measures would be contradicted).
  Note that $\mu$ is not necessarily translation-invariant, but translating it by $g$ gives another $F$-invariant ergodic measure so we can actually suppose ${g=\cspac}$ so ${\mu([q])<1}$.
  Then by the ergodic Theorem (see \cite[Lemma 6.13]{walters1981}), we would have some $x\in\confs$ with 
  \[\frac{1}{n}\sum_{t=0}^{n-1}f\bigl(F^t(x)\bigr)\to_n \mu([q])\]
  where $f$ is the characteristic map of $[q]$. But then ${d_q(\trac{F}{x})<1}$ which contradicts the hypothesis.  
\end{proof}

Uniquely ergodic CA form a strictly larger class than nilpotent CA.
The example in the following result to show this separation is highly non-trivial and we will not present it here, but it is quite remarkable that such behaviors are possible with CA.

\begin{theorem}[\cite{TORMA2015415}]
  \label{theo:nilpernotue}
  There exists a uniquely ergodic CA on $\Z$ which is not nilpotent. 
\end{theorem}

\begin{remark}
  The CA from the above theorem cannot be convergent because then it would be asymptotically nilpotent by Lemma~\ref{lem:traceuniqergo} and this would contradict Theorem~\ref{theo:asnilimpnil}.
  From Lemma~\ref{lem:convergencemeasure} it is also non-convergent when starting from some measures.
  However, there is convergence in Cesaro mean starting from any ${\mu}$ since otherwise, a Cesaro mean sequence would have several limit points and hence produce several invariant measure for the CA as already explained above.
\end{remark}

\begin{theorem}
  There exists a CA on $\Z$ which is nilpotent over periodic configurations but not uniquely ergodic.
\end{theorem}
\begin{proof}[Sketch of proof]
  Such a CA can be obtained using the classical construction of \cite{kari92} using Wang tiles, which formalizes some tiling of the plane $\Z^2$.
  Roughly, a Wang tileset is a finite set of tile types $X$ together with vertical and horizontal constraints specifying which pairs of adjacent tile types is allowed horizontally (resp. vertically).
  A valid tiling is a configuration of ${X^{\Z^2}}$ that satisfies the constraints everywhere.
  A tileset is aperiodic if it admits valid tilings, but no periodic ones.
  A tileset is NE-deterministic when the constraints have the following property: when fixing a tile at the north of a position and another one at the east, then there is at most one possible tile type that can be placed at the north-east without violating the local horizontal and vertical constraints.
  Taking a NE-deterministic aperiodic Wang tileset $T$ (whose existence is proven in \cite{kari92}) one can transform it into a one-dimensional CA $F$ with alphabet ${T\cup\{s\}}$, where $s$ is a spreading state, and neighborhood ${\{0,1\}}$ as follows: if a cell and its right neighbor hold a pair of state from $T$ that represent two tiles (one at the north, one at the east) that can be uniquely completed to $t$ according to the NE-deterministic rule, then the cell turns into state $t$, otherwise it turns into state $s$ (configurations of the CA represent diagonals of a tiling).
  One obtains a CA $F$ that admits ${\delta_{\overline{s}}}$ as invariant measure.
  Because the tileset has no valid periodic tiling, $F$ is nilpotent on periodic configurations (the orbit of a periodic configuration without occurrence of $s$ would give a periodic tiling of $T$, which is impossible, and any occurrence of $s$ in a periodic configuration forces its orbit to converge towards $\overline{s}$). 
  On the other hand, because the tileset admits a valid (aperiodic) tiling, it means that there is a configuration $c$ whose orbit under $F$ has no occurrence of $s$.
  Taking any limit point of the sequence of Cesaro means 
  \[\mu_n=\frac{1}{n}\sum_{t=0}^{n-1}F^t(\delta_c)\]
  gives an invariant measure such that ${\mu([s])=0}$ since ${\mu_n([s])=0}$ for all ${n\in\N}$.
  We thus get a second invariant measure for $F$.
\end{proof}

So far we considered properties concerning all initial configurations, or a meager set of initial configuration.
Natural variants of nilpotency have been considered in the literature to capture the behavior on \emph{most} or \emph{a large set} of initial configurations.
Two notions in particular emerged in literature that use different approach to define precisely \emph{most} or \emph{large set}: one is topological \cite{Milnor_1985,DjenaouiG19}, the other uses measure theory \cite{mulimset}.
Recall that a set in a topological space is comeager (intuitively, large) if it is a countable intersection of sets with dense topological interior (for instance a countable intersection of dense open sets).

\begin{definition}
  Consider a CA $F$ on $\confs$ and ${\mu\in\meas{\confs}}$. Then $F$ is:
  \begin{itemize}
  \item \textit{generically nilpotent} if there is ${q\in Q}$ and a comeager set ${B\subseteq\confs}$ such that ${F^t(x)\to_t\overline{q}}$ for all ${x\in B}$;
  \item \textit{$\mu$-nilpotent} if there is ${q\in Q}$ such that ${F^t(\mu)\to_t\delta_{\overline{q}}}$.
  \end{itemize}
\end{definition}

Generic nilpotency and $\mu$-nilpotency both specify a behavior on \emph{most configurations}, where \emph{most} is to be understood either topologically or according to a  measure.
It turns out that for well-behaved measures the topological version is stronger.

\begin{theorem}[\cite{DjenaouiG19}]\label{theo:genmunil}
  Let $F$ be a generically nilpotent CA on $\confs$ and $\mu\in\meas{\confs}$ a full-support translation-ergodic measure. Then $F$ is $\mu$-nilpotent.
\end{theorem}

\begin{figure}
  \label{fig:nildiag}
  \centering
  \begin{tikzpicture}[scale=.7,node distance=1cm,>=stealth]
    \tikzstyle{nilprop}=[draw, rectangle]
    \fill[fill=gray,path fading=east] (-2,1.25) rectangle (3,1.3);
    \fill[fill=gray,path fading=east] (-2,-1.25) rectangle (3,-1.3);
    \draw (0,2.5) node {\small\textit{\begin{tabular}{c}Dense subset of\\ configurations
                                \end{tabular}
                              }};
    \draw (0,0) node {\small\textit{\begin{tabular}{c}All\\ configurations
                              \end{tabular}
                            }};
    \draw (0,-2.5) node {\small\textit{\begin{tabular}{c}'Most'\\ configurations
                                 \end{tabular}
                               }};
    \draw (12.5,2.5) node[nilprop] {\small\begin{tabular}{c}nilpotent over\\ periodic\\ configurations
                           \end{tabular}
                         };
    \draw (8,2.5) node[nilprop] {\small\begin{tabular}{c}nilpotent over\\ finite\\ configurations
                           \end{tabular}
                         };
    \draw[->,very thick] (7.75,.75) -- (7.75,1.55);
    \draw[->,very thick] (8.25,1.55) -- node[thin,cross out,draw=black,minimum size=2mm] {} (8.25,.75);
    \draw[->,very thick] (9.85,2.5) -- node[thin,cross out,draw=black,minimum size=2mm] {} (10.65,2.5);
    \draw[->,very thick,dotted] (12.25,.75) -- node[midway] (UEnilper) {} (12.25,1.55);
    \draw[->,very thick] (12.75,1.55) -- node[pos=0.4] (Nilpernoteu) {} (12.75,.75);
    \draw (UEnilper)++(-.3,0) node[left] {\tiny\it Lem.~\ref{lem:uniqergonilperio}};
    \draw (Nilpernoteu) node[cross out,draw=black,minimum size=2mm] {};
    \draw (Nilpernoteu)++(.3,0) node[right] {\tiny\it Thm.~\ref{theo:nilpernotue}};
    \draw (3,0) node[nilprop] {nilpotent};
    \draw[->,very thick] (4.3,.25) -- (6,.25);
    \draw[->,very thick,dotted] (6,-.25) -- node[midway] (Nil) {} (4.3,-.25);
    \draw (Nil)++(0,-.3) node {\tiny\it Thm.~\ref{theo:asnilimpnil}};
    \draw[->,very thick] (11,-.25) -- node[midway] (UEnotnil) {} (10,-.25);
    \draw[->,very thick] (10,.25) -- node[midway,above] {\tiny\it Lem.~\ref{lem:asynilmeasure}} (11,.25);
    \draw (UEnotnil) node[cross out,draw=black,minimum size=2mm] {};
    \draw (UEnotnil)++(0,-.3) node {\tiny\it Thm.~\ref{theo:nilpernotue}};
    \draw (8,0) node[nilprop] {\small\begin{tabular}{c}asymptotically\\ nilpotent
                                 \end{tabular}
                               };
    \draw (12.5,0) node[nilprop] {\small\begin{tabular}{c}uniquely\\ ergodic
                                  \end{tabular}
                                };
    \draw (8,-2.5) node[nilprop] {\small\begin{tabular}{c}generically\\ nilpotent
                                    \end{tabular}
                                  };
    \draw (12.5,-2.5) node[nilprop] {\small $\mu$-nilpotent};
    \draw[->,very thick] (7.75,-0.75) -- (7.75,-1.55);
    \draw[->,very thick] (8.25,-1.55) -- node[pos=0.4] (Gennoasnil) {} (8.25,-0.75);
    \draw (Gennoasnil) node[cross out,draw=black,minimum size=2mm] {};
    \draw[->,very thick,dotted] (10,-2.25) -- node[midway] (Genmunil) {} (11,-2.25);
    \draw (Genmunil)++(0,.3) node {\tiny\it Thm.~\ref{theo:genmunil}};
    \draw[->,very thick] (11,-2.75) -- node[midway] (Munotgen) {} (10,-2.75);
    \draw (Munotgen) node[cross out,draw=black,minimum size=2mm] {};
    \draw (Munotgen)++(0,-.3) node {\tiny\it Thm.~\ref{theo:classicbootstrap}};
  \end{tikzpicture}
  \caption{Nilpotency and its variants with (non-)implications between them as seen so far. Implications with dotted arrows refer to Theorems that require some hypothesis on either $\spac$ or $\mu$. Counter-examples (or non-implications) are represented with a cross on the arrow. }
  \label{fig:nilimp}
\end{figure}
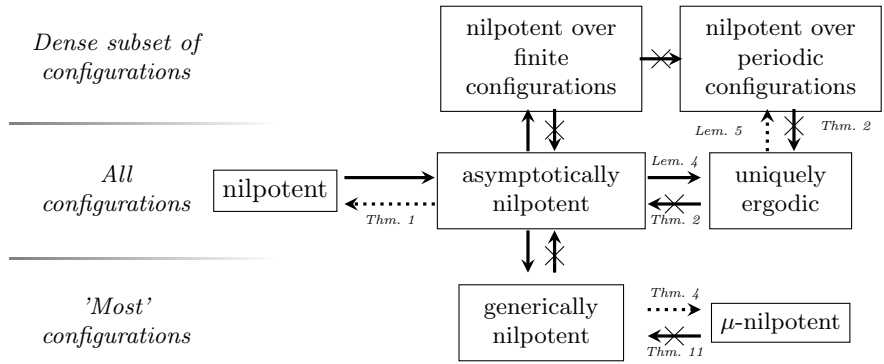

There are ways to characterize the previous variants of nilpotency by defining variants of limit sets, \textit{i.e.} sets of (typical) asymptotic configurations \cite{mulimset,DjenaouiG19}.
Each time, the nilpotency variant is equivalent to having such asymptotic sets being singletons (see \cite{torma2021generically,Boyer_2015}).

The implication diagram of Figure~\ref{fig:nilimp} is not complete, in particular concerning the link between unique ergodicity and $\mu$-nilpotency.

\begin{question}
  What are the other implications or counter-examples in the diagram of Figure~\ref{fig:nilimp}?
  In particular, for what measure $\mu$ does unique ergodicity imply $\mu$-nilpotency? 
  Said differently, for what measure $\mu$ the convergence in Cesaro mean implied by unique ergodicity is actually a simple convergence?
\end{question}

\subsection{The undecidability garden}

We conclude this section by undecidability results related to the notions presented above.
Of course, an undecidability result is a negative result and this might appear repulsive to some.
However, behind almost any undecidability result there is a positive one: a construction technique (a reduction) providing examples with rich and complex behaviors.
Each of the results below relies on a non-trivial and nice construction idea that is of independent interest.
A presentation of each technique would take too much space for this tutorial, but we invite the reader to dig into the corresponding references.

The following theorem follows from the undecidability of the domino problem on deterministic sets of Wang tiles. Several proof techniques are known and we suggest to read \cite[Chapter 6]{bookchairmorlet2020} for a survey.

\begin{theorem}[\cite{kari92,aanderaa_lewis_1974}]
  The set of 1D nilpotent CA is ${\Sigma_1^0}$-complete.
\end{theorem}

As for the existence of a uniquely ergodic CA which is not nilpotent (Theorem~\ref{theo:nilpernotue}), the key tool in the following result is a self-simulating CA inspired from \cite{G_cs_2001} in which the only orbits than can survive produce sparser and sparser computations.

\begin{theorem}[\cite{TORMA2015415}]
  The set of 1D uniquely ergodic CA is $\Pi_2^0$-complete.
\end{theorem}

The two following results use a construction technique introduced in \cite{Delacourt_2011} which has proven to be useful when one wants to control the behavior of a CA on most initial configuration: it relies on a special precursor state that generates well-behaved zones that grow in a random context until they meet another well-behaved zone.
Each zone can hold a properly initialized Turing computation.
Then, a global process of merging or destruction of meeting zones is used to ensure that a least some Turing computations can go on forever.

\begin{theorem}[\cite{torma2021generically}]
  The set of 1D generically nilpotent CA is $\Sigma_2^0$-complete.
\end{theorem}

Recall from Theorem~\ref{theo:genmunil} that a generically nilpotent CA is necessarily $\mu$-nilpotent for $\mu$ the uniform Bernoulli measure (which is ergodic and has full support).
The converse is false and the following theorem illustrates how the weaker property of $\mu$-nilpotence is strictly harder to decide.

\begin{theorem}[\cite{Boyer_2015}]
  Let $\mu$ denote the uniform Bernoulli measure. The set of 1D $\mu$-nilpotent CA is $\Pi_3^0$-complete.
\end{theorem}

\section{Proving Convergence}
\label{sec:proveconv}

\subsection{Bounded change CA}

Proving convergence of a CA $F$ means proving that there exists a time bound for each trace ${\trac{F}{c}}$ after which the trace is constant.
This has been defined above as the freezing time ${\freezt{F}{c}{0}}$.
In general this time bound depends on the initial configuration (see Example~\ref{exa:zigzag}), but, even if there is a global bound that doesn't depend on initial configuration like in nilpotent CA, it is generally uncomputable.
Instead of bounding directly the time of the last change in traces, a way to tackle the problem is to bound the number of changes that can occur in a trace.
The following definition is useful \cite{vollmar81,KutribM10a}.

\begin{definition}\label{def:boundedchange}
  A CA $F$ is $k$-change for some ${k\in\N}$ if for any ${x\in\confs}$, there are at most $k$ state changes in ${\trac{F}{x}}$, \textit{i.e.} 
  \[\#\{t\in\N : \trac{F}{x}(t+1)\neq\trac{F}{x}(t)\}\leq k.\]
  A CA is \emph{bounded-change} if it is $k$-change for some $k$.
\end{definition}

A bounded-change CA is necessarily convergent: a trace which is not ultimately constant cannot exist since it would contain infinitely many changes.
However the converse is false as shown by the following example (it is used extensively as a basic building block in \cite{ollinger:hal-02266916}).

\renewcommand\stateZ[2]{\draw[fill=white] (#1,#2) rectangle +(1,1); \draw (#1,#2)+(.5,.5) node {\tiny $B$};}
\renewcommand\stateO[2]{\draw[fill=white] (#1,#2) rectangle +(1,1); \draw (#1,#2)+(.5,.5) node {\tiny $B$};}
\renewcommand\stateZO[2]{\draw[fill=white!50!gray] (#1,#2) rectangle +(1,1);\draw (#1,#2)+(.5,.5) node {\tiny $\to$};}
\newcommand\stateOO[2]{\draw[fill=gray] (#1,#2) rectangle +(1,1);\draw (#1,#2)+(.5,.5) node {\tiny $\leftarrow$};}
\newcommand\stateZZO[2]{\draw[fill=white!50!red] (#1,#2) rectangle +(1,1); \draw (#1,#2)+(.5,.5) node {\tiny $<$};}
\newcommand\stateOZO[2]{\draw[fill=white!50!blue] (#1,#2) rectangle +(1,1); \draw (#1,#2)+(.5,.5) node {\tiny $>$};}
\newcommand\stateZOO[2]{\draw[fill=white] (#1,#2) rectangle +(1,1); \draw (#1,#2)+(.5,.5) node {\tiny $B$};}
\newcommand\stateOOO[2]{\draw[fill=white] (#1,#2) rectangle +(1,1); \draw (#1,#2)+(.5,.5) node {\tiny $B$};}

\newcommand\etatblanc{\tikz[baseline,scale=.3]{\stateZ{0}{0}}}
\newcommand\etatagauche{\tikz[scale=.3]{\stateZO{0}{0}}}
\newcommand\etatadroite{\tikz[scale=.3]{\stateOO{0}{0}}}
\newcommand\etattetedroite{\tikz[baseline,scale=.3]{\stateOZO{0}{0}}}
\newcommand\etattetegauche{\tikz[baseline,scale=.3]{\stateZZO{0}{0}}}

\begin{example}\label{exa:zigzag}
  Let us consider the state set ${Q=\{\etatblanc,\etatagauche,\etatadroite,\etattetegauche,\etattetedroite,e\}}$ where $\etattetegauche$ and $\etattetedroite$ represent an active head moving left or right respectively, $\etatagauche$ and $\etatadroite$ are states indicating the position of the head (to the right and to the left respectively), $\etatblanc$ is a blank state and $e$ is an error state. Let $\Sigma$ be the set of configurations whose patterns of length two all appear in the space-time diagram of Figure~\ref{fig:convnotbc}, \textit{i.e.} configurations made of zones with a unique active head (possibly none if the zone is infinite) in a blank background. Define the CA ${F:Q^\Z\to Q^\Z}$ of radius $2$, such that:
  \begin{itemize}
  \item $e$ is a spreading state and any cell turns into state $e$ if there is some pattern not in $L_2$ in its neighborhood,
  \item in the $\Sigma$-valid zones with a unique head, the only cells that change their states in one step are the one around the head, and they update in such a way that the head makes zigzag inside the zone and reduces the size of the zone by one unit at each bounce on a border. Precisely, all the transitions appear in the space-time diagram of Figure~\ref{fig:convnotbc}.
  \end{itemize}
  First, it is clear that $F$ cannot be $k$-change whatever the value of $k$ is, because in a $\Sigma$-valid finite zone of size $n$, the orbit of the central cells contains at least $n$ changes.
  On the other hand, if we consider any initial configuration ${c\in Q^\Z}$ and any position ${z\in\Z}$ there are only $3$ possible cases:
  \begin{itemize}
  \item either ${c\not\in\Sigma}$ and in this case the spreading state $e$ will appear somewhere and reach position $z$ at some step,
  \item or ${c\in\Sigma}$ and $z$ belongs to a finite valid zone, so after some number of zigzags of the head, the zone will shrink letting $z$ outside in state $\etatblanc$,
  \item or ${c\in\Sigma}$ and $z$ belongs to an infinite valid zone, and after at most two passages of the head (at most one bounce on the boundary if any), the head will move towards infinity and cell $z$ will remain unchanged forever.
  \end{itemize}
  In any case, we have that the trace at $z$ starting from $c$ is eventually constant. Hence $F$ is convergent.
\end{example}

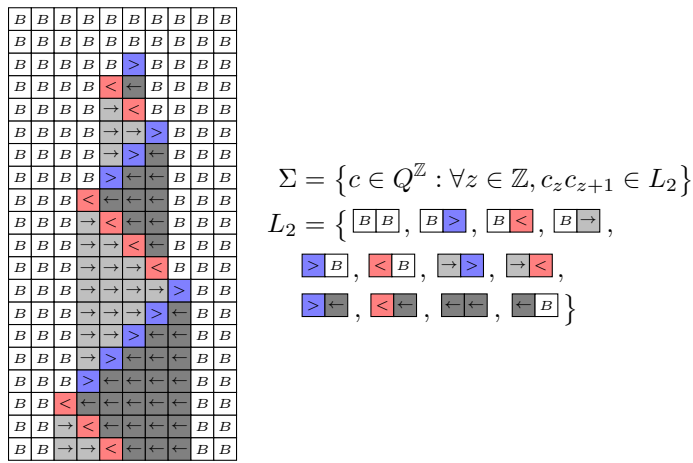
\begin{figure}
  \begin{center}
    \begin{tikzpicture}[scale=.3,baseline=(current bounding box.center)]
      \stateZ{5}{0}\stateZ{5}{1}\stateZ{5}{2}\stateZ{5}{3}\stateZ{5}{4}\stateZ{5}{5}\stateZ{5}{6}\stateZ{5}{7}\stateZ{5}{8}\stateZ{5}{9}\stateZ{5}{10}\stateZ{5}{11}\stateZ{5}{12}\stateZ{5}{13}\stateZ{5}{14}\stateZ{5}{15}\stateZ{5}{16}\stateZ{5}{17}\stateZ{5}{18}\stateZ{5}{19}
      \stateO{6}{0}\stateO{6}{1}\stateO{6}{2}\stateZ{6}{3}\stateZ{6}{4}\stateZ{6}{5}\stateZ{6}{6}\stateZ{6}{7}\stateZ{6}{8}\stateZ{6}{9}\stateZ{6}{10}\stateZ{6}{11}\stateZ{6}{12}\stateZ{6}{13}\stateZ{6}{14}\stateZ{6}{15}\stateZ{6}{16}\stateZ{6}{17}\stateZ{6}{18}\stateZ{6}{19}
      \stateZO{7}{0}\stateZO{7}{1}\stateZZO{7}{2}\stateO{7}{3}\stateO{7}{4}\stateO{7}{5}\stateO{7}{6}\stateO{7}{7}\stateO{7}{8}\stateO{7}{9}\stateO{7}{10}\stateO{7}{11}\stateZ{7}{12}\stateZ{7}{13}\stateZ{7}{14}\stateZ{7}{15}\stateZ{7}{16}\stateZ{7}{17}\stateZ{7}{18}\stateZ{7}{19}
      \stateZO{8}{0}\stateZZO{8}{1}\stateOO{8}{2}\stateOZO{8}{3}\stateZO{8}{4}\stateZO{8}{5}\stateZO{8}{6}\stateZO{8}{7}\stateZO{8}{8}\stateZO{8}{9}\stateZO{8}{10}\stateZZO{8}{11}\stateO{8}{12}\stateO{8}{13}\stateO{8}{14}\stateO{8}{15}\stateO{8}{16}\stateZ{8}{17}\stateZ{8}{18}\stateZ{8}{19}
      \stateZZO{9}{0}\stateOO{9}{1}\stateOO{9}{2}\stateOO{9}{3}\stateOZO{9}{4}\stateZO{9}{5}\stateZO{9}{6}\stateZO{9}{7}\stateZO{9}{8}\stateZO{9}{9}\stateZZO{9}{10}\stateOO{9}{11}\stateOZO{9}{12}\stateZO{9}{13}\stateZO{9}{14}\stateZO{9}{15}\stateZZO{9}{16}\stateO{9}{17}\stateO{9}{18}\stateO{9}{19}
      \stateOO{10}{0}\stateOO{10}{1}\stateOO{10}{2}\stateOO{10}{3}\stateOO{10}{4}\stateOZO{10}{5}\stateZO{10}{6}\stateZO{10}{7}\stateZO{10}{8}\stateZZO{10}{9}\stateOO{10}{10}\stateOO{10}{11}\stateOO{10}{12}\stateOZO{10}{13}\stateZO{10}{14}\stateZZO{10}{15}\stateOO{10}{16}\stateOZO{10}{17}\stateZOO{10}{18}\stateZOO{10}{19}
      \stateOO{11}{0}\stateOO{11}{1}\stateOO{11}{2}\stateOO{11}{3}\stateOO{11}{4}\stateOO{11}{5}\stateOZO{11}{6}\stateZO{11}{7}\stateZZO{11}{8}\stateOO{11}{9}\stateOO{11}{10}\stateOO{11}{11}\stateOO{11}{12}\stateOO{11}{13}\stateOZO{11}{14}\stateZOO{11}{15}\stateZOO{11}{16}\stateZOO{11}{17}\stateOOO{11}{18}\stateOOO{11}{19}
      \stateOO{12}{0}\stateOO{12}{1}\stateOO{12}{2}\stateOO{12}{3}\stateOO{12}{4}\stateOO{12}{5}\stateOO{12}{6}\stateOZO{12}{7}\stateZOO{12}{8}\stateZOO{12}{9}\stateZOO{12}{10}\stateZOO{12}{11}\stateZOO{12}{12}\stateZOO{12}{13}\stateZOO{12}{14}\stateOOO{12}{15}\stateOOO{12}{16}\stateOOO{12}{17}\stateOOO{12}{18}\stateOOO{12}{19}
      \stateZOO{13}{0}\stateZOO{13}{1}\stateZOO{13}{2}\stateZOO{13}{3}\stateZOO{13}{4}\stateZOO{13}{5}\stateZOO{13}{6}\stateZOO{13}{7}\stateOOO{13}{8}\stateOOO{13}{9}\stateOOO{13}{10}\stateOOO{13}{11}\stateOOO{13}{12}\stateOOO{13}{13}\stateOOO{13}{14}\stateOOO{13}{15}\stateOOO{13}{16}\stateOOO{13}{17}\stateOOO{13}{18}\stateOOO{13}{19}
      \stateOOO{14}{0}\stateOOO{14}{1}\stateOOO{14}{2}\stateOOO{14}{3}\stateOOO{14}{4}\stateOOO{14}{5}\stateOOO{14}{6}\stateOOO{14}{7}\stateOOO{14}{8}\stateOOO{14}{9}\stateOOO{14}{10}\stateOOO{14}{11}\stateOOO{14}{12}\stateOOO{14}{13}\stateOOO{14}{14}\stateOOO{14}{15}\stateOOO{14}{16}\stateOOO{14}{17}\stateOOO{14}{18}\stateOOO{14}{19}
    \end{tikzpicture}
    \begin{minipage}[c]{0.5\linewidth}
      \begin{center}
      \begin{align*}
        \Sigma &= \bigl\{c\in Q^\Z : \forall z\in\Z, c_zc_{z+1}\in L_2\bigr\}\\
        L_2 &= \bigl\{\tikz[baseline,scale=.3]{\stateZ{0}{0}\stateZ{1}{0}},\tikz[baseline,scale=.3]{\stateZ{0}{0}\stateOZO{1}{0}},\tikz[baseline,scale=.3]{\stateZ{0}{0}\stateZZO{1}{0}},\tikz[baseline,scale=.3]{\stateZ{0}{0}\stateZO{1}{0}},\\
               &\tikz[baseline,scale=.3]{\stateOZO{0}{0}\stateZ{1}{0}},\tikz[baseline,scale=.3]{\stateZZO{0}{0}\stateZ{1}{0}},\tikz[baseline,scale=.3]{\stateZO{0}{0}\stateOZO{1}{0}},\tikz[baseline,scale=.3]{\stateZO{0}{0}\stateZZO{1}{0}},\\
               & \tikz[baseline,scale=.3]{\stateOZO{0}{0}\stateOO{1}{0}},\tikz[baseline,scale=.3]{\stateZZO{0}{0}\stateOO{1}{0}},\tikz[baseline,scale=.3]{\stateOO{0}{0}\stateOO{1}{0}},\tikz[baseline,scale=.3]{\stateOO{0}{0}\stateZ{1}{0}}\bigr\}
      \end{align*}
    \end{center}
    \end{minipage}
    \end{center}
    \caption{A convergent but not bounded-change CA (time from bottom to top in the space-time diagram on the left).}
    \label{fig:convnotbc}
  \end{figure}
To use a metaphor inspired from physics, each cell of a bounded-change CA can be seen as a system that evolves according to information received from its neighbors, and for which each state change has a fixed energy cost.
Then, in each orbit and at each cell, the instantaneous energy defined as the number of remaining potential state changes is non-increasing with time, and strictly decreasing at each state change.
The number $k$ in Definition~\ref{def:boundedchange} can be interpreted as the cell-wise capacity, \textit{i.e.} the maximal energy an individual cell can hold initially.

As seen in the previous section, a nilpotent CA $F$ is such that ${F^t}$ is a constant map for some $t$, so it is always bounded-change.
Nilpotency for CA with a spreading state is actually Turing reducible to the property of bounded-change as follows: consider $F$ with a spreading state and add a $\{0,1\}$ component to states that constantly exchanges $0$ and $1$, except if there is a spreading state in the neighborhood in which case the additional component turns into $0$.
This new CA $F'$ is bounded-change if and only if $F$ is nilpotent: first, if $F$ is nilpotent then $F'$ also because the spreading state of $F$ forces $0$ everywhere on the second component; conversely, if $F$ is not nilpotent then it must possess an orbit without any occurrence of the spreading state (by compacity), and then the corresponding orbit in $F'$ is constantly changing the state on the second layer.
By undecidability of nilpotency we deduce the following theorem.

\begin{theorem}
  \label{theo:bcundecidable}
  It is undecidable to determine whether a 1D CA is bounded-change or not.
\end{theorem}

Despite the general undecidability, this approach turns out to be very fruitful to analyze some majority CA.
A majority cellular automaton consists in taking in each cell the state which has the majority of occurrences among the neighboring cells.
The majority CA with neighborhood $V\subseteq\spac$ is defined on alphabet ${Q=\{0,1\}}$ by 
\[F(x)_z =
  \begin{cases}
    1 & \text{ if $\displaystyle\sum_{i\in V} x_{z+i} >\#V/2$, or $\displaystyle\sum_{i\in V} x_{z+i} =\#V/2$ and $x_z=1$}\\
    0 & \text{ otherwise.}
  \end{cases}
\]

A subset ${S\subseteq\Z^d}$ is symmetric if ${x\in S\implies -x\in S}$.
The following result concerning majority CA is a specific case of a much more general result from \cite{majoritybc2000} which deals with infinite graphs with some growth condition.
A similar qualitative behavior for majority evolution rules was established before in the settings of automata networks (the lattice of cells is finite, the evolution rule is possibly non-uniform), see \cite{PaperGoles}.
It is based on a decreasing energy function in both case.
Note however that \cite{majoritybc2000} do not state their main theorem as a bounded change property but only as a convergence result, although their proof technique allows to establish bounded-changedness as shown below.

\begin{figure}
  \centering
  \begin{minipage}{.2\linewidth}
    \begin{center}
      \includegraphics[width=.9\textwidth]{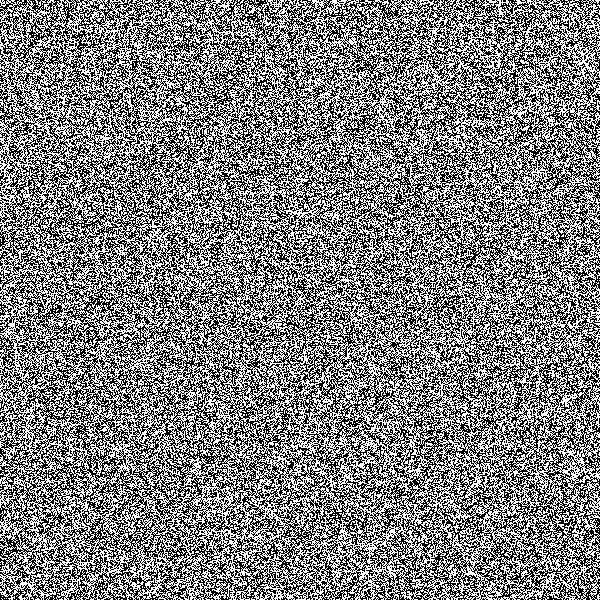}\\
      $c$
    \end{center}
  \end{minipage}
  \begin{minipage}{.2\linewidth}
    \begin{center}
      \includegraphics[width=.9\textwidth]{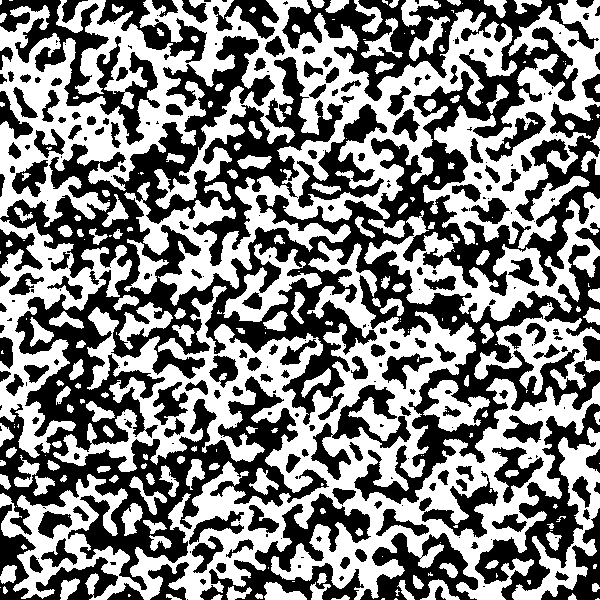}\\
      $F^2(c)$
    \end{center}
  \end{minipage}
  \begin{minipage}{.2\linewidth}
    \begin{center}
      \includegraphics[width=.9\textwidth]{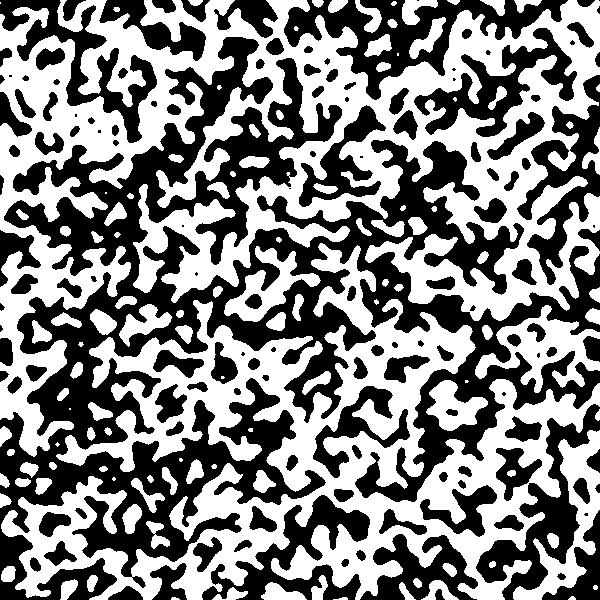}\\
      $F^4(c)$
    \end{center}
  \end{minipage}
  \begin{minipage}{.2\linewidth}
    \begin{center}
      \includegraphics[width=.9\textwidth]{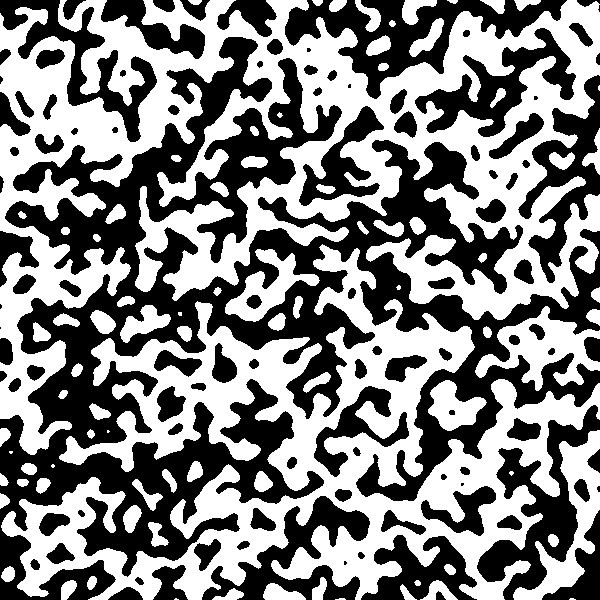}\\
      $F^6(c)$
    \end{center}
  \end{minipage}
  \vskip 2em
  \begin{minipage}{.2\linewidth}
    \begin{center}
      \includegraphics[width=.9\textwidth]{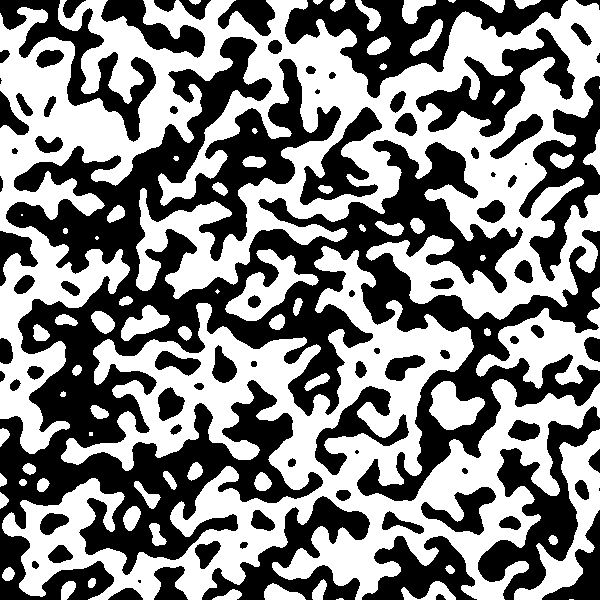}\\
      $F^8(c)$
    \end{center}
  \end{minipage}
  \hskip 1cm $\cdots$\hskip 1cm
  \begin{minipage}{.2\linewidth}
    \begin{center}
      \includegraphics[width=.9\textwidth]{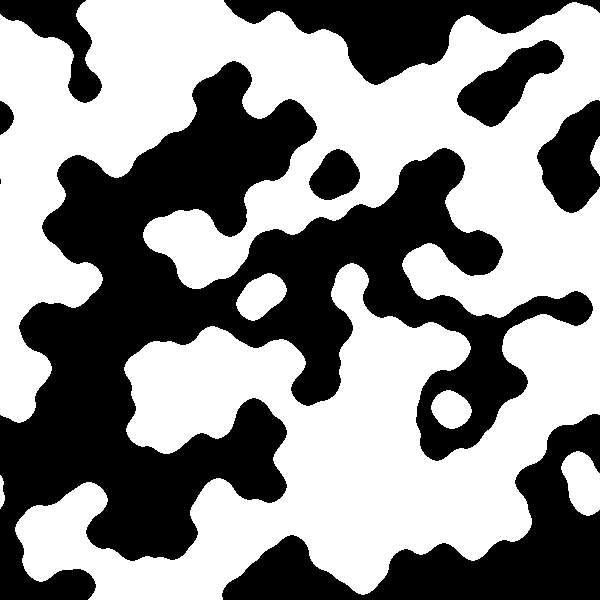}\\
      $\bigl(F^2\bigr)^\omega(c)$
    \end{center}
  \end{minipage}
  \caption{An orbit of $F^2$ where $F$ is some majority CA on $\Z^2$ with symmetric neighborhood $S$ with $(0,0)\in S$ as in the conditions of Theorem~\ref{theo:majboundedchange}, so that $F^2$ is bounded-change and hence convergent.}
  \label{fig:maj2D}
\end{figure}

\begin{theorem}[\cite{majoritybc2000}]
  \label{theo:majboundedchange}
  Let $F$ be any majority CA with ${\spac=\Z^d}$ and symmetric neighborhood $S$ with ${0\in S}$, then $F^2$ is bounded-change.
\end{theorem}
\begin{proof}
  Denote by $r$ the radius of the CA, \textit{i.e.} the maximum ${\|u\|_\infty}$ for ${u\in S}$.
  Let us fix any ${x\in\confs}$.
  The core of the argument is to consider the following ``energy'' function for any ${x\in\confs}$ and any ${n,t\in\N}$: 
  \[E_n^t = \sum_{u\in\boule{n}}\sum_{v\in\boule{n}}a(u,v)\bigl|F^{t+1}(x)_u-F^t(x)_v\bigr|\]
  where ${a(u,v)=1}$ if ${v\in u+S}$, and $0$ else.
  Because $S$ is symmetric ${a(u,v)=a(v,u)}$ so we can write $E^{t-1}_n$ as follows:
  \[E_n^{t-1} = \sum_{u\in\boule{n}}\sum_{v\in\boule{n}}a(u,v)\bigl|F^{t-1}(x)_u-F^t(x)_v\bigr|.\]  
  For any ${t>0}$, we will first give an upper bound on ${\Delta_n^t = E_n^{t}-E_n^{t-1}}$ which intuitively means that the energy $E_n^t$ essentially decreases with time, but only up to a ``controlled perturbation'' of the border of the ball $\boule{n}$. Denoting 
  \[\Sigma_{n,t}^+(u) = \sum_{v\in\boule{n}}a(u,v)\bigl|F^{t+1}(x)_u-F^t(x)_v\bigr|,\text{ and}\]
  \[\Sigma_{n,t}^-(u) = \sum_{v\in\boule{n}}a(u,v)\bigl|F^{t-1}(x)_u-F^t(x)_v\bigr|\]
  we can rewrite $\Delta_n^t$ as 
  \begin{equation}
    \Delta_n^t = \sum_{0\leq i\leq n-r}\sum_{u\in C_{i,t}}\Sigma_{n,t}^+-\Sigma_{n,t}^- + \sum_{n-r< i\leq n}\sum_{u\in C_{i,t}}\Sigma_{n,t}^+-\Sigma_{n,t}^- \label{eq:deltatwosums}    
  \end{equation}
  where $C_{i,t}$ denote the cells that have different states between step $t+1$ and $t-1$ (\textit{i.e.} ${F^{t+1}(x)_u\neq F^{t-1}(x)_u}$) and belong to the sphere of radius $i$ (\textit{i.e.} ${\|u\|_\infty=i}$).
  This expression for $\Delta_n^t$ holds because ${\Sigma_{n,t}^+(u)=\Sigma_{n,t}^-(u)}$ when ${F^{t+1}(x)_u= F^{t-1}(x)_u}$.
  
  We first claim that ${\Sigma_{n,t}^+(u)-\Sigma_{n,t}^-(u)\leq \frac{\#S -1}{2}}$ for any $u\in\boule{n}$: this holds because actually ${\Sigma_{n,t}^+(u)\leq \frac{\#S -1}{2}}$ must hold to ensure that a majority of neighbors of $u$ are in state ${F^{t+1}(x)_u}$ at time $t$.
  
  We now claim that ${\Sigma_{n,t}^+(u)-\Sigma_{n,t}^-(u)\leq -1}$ for any ${\displaystyle u\in\bigcup_{0\leq i\leq n-r}C_{i,t}}$.
  Indeed, we can split the neighbors of any such $u$ into two sets 
  \[S_u^+ = \{v\in u+S : F^t(x)_v=F^{t+1}(x)_u\},\text{ and }\]
  \[S_u^- = \{v\in u+S : F^t(x)_v=F^{t-1}(x)_u\},\]
  and then ${\# S_u^+= \Sigma_{n,t}^-(u)}$ because ${u+S\subseteq\boule{n}}$, and similarly ${\# S_u^-= \Sigma_{n,t}^+(u)}$.
  By the majority rule we must have ${\# S_u^+>\# S_u^-}$ (recall that $S$ is of odd cardinality by hypothesis) which proves the claim.
  
  Let ${s=\frac{(\#S -1)}{2}}$. From the two claims we can rewrite Equation~(\ref{eq:deltatwosums}) above as:
  \begin{equation}
    \label{eq:deltafinal}\displaystyle
    \Delta_n^t\leq - \#\bigcup_{0\leq i\leq n-r}C_{i,t} + s\#\bigcup_{n-r< i\leq n}C_{i,t}.
  \end{equation}
  We are now going to deduce a bound on the number of changes in ${\trac{F^2}{x}}$ from the above inequality in two steps: using a weighted sum spatially and then summing over time. 
  Spatially, the trick is to use a telescoping sum to attenuate the perturbation of the second term in the upper bound of Equation~(\ref{eq:deltafinal}) and put more weight on the central cell.
  Choose ${n = kr}$ and consider the sum
  \[TE_n^t = E_n^t + \frac{1}{s} E_{n-r}^t + (1+\frac{1}{s})\frac{1}{s} E_{n-2r}^t + (1+\frac{1}{s})^2\frac{1}{s} E_{n-3r}^t + \cdots + (1+\frac{1}{s})^{k-1}\frac{1}{s} E_0^t.\]
  The weights are chosen so that ${TE_n^t-TE_n^{t-1} = \Delta_n^t + \frac{1}{s} \Delta_{n-r}^t + \cdots}$ can be nicely upper-bounded.
  Indeed Equation~(\ref{eq:deltafinal}) gives \[\Delta_n^t + \frac{1}{s} \Delta_{n-r}^t\leq -(1+\frac{1}{s})\#\bigcup_{0\leq i\leq n-2r}C_{i,t} + s\#\bigcup_{n-r< i\leq n}C_{i,t}\] so the changes in the annulus between spheres of radii ${n-2r}$ and ${n-r}$ are exactly compensated.
  Then, when adding ${(1+\frac{1}{s})\frac{1}{s} \Delta_{n-2r}^t}$ upper-bounded by Equation~(\ref{eq:deltafinal}), the terms involving annulus from ${n-3r}$ to ${n-2r}$ are exactly compensated again, giving: 
  \[\Delta_n^t + \frac{1}{s} \Delta_{n-r}^t + (1+\frac{1}{s})\frac{1}{s} \Delta_{n-2r}^t \leq -(1+\frac{1}{s})^2\#\bigcup_{0\leq i\leq n-3r}C_{i,t} + s\#\bigcup_{n-r< i\leq n}C_{i,t}.\]
  Repeating this argument, and by the identity 
  \[1 + \frac{1}{s} + (1+\frac{1}{s})\frac{1}{s} + \cdots + (1+\frac{1}{s})^{k-1}\frac{1}{s} = (1+\frac{1}{s})^k,\]
  we get 
  \[TE_n^t - TE_n^{t-1}\leq - (1+\frac{1}{s})^kC_{0,t} + s\#\bigcup_{n-r< i\leq n}C_{i,t} \leq - (1+\frac{1}{s})^kC_{0,t} + s\#(\boule{n}\setminus\boule{n-r}).\]
  By summing over time the successive differences we then obtain 
  \[TE_n^t-TE_n^0\leq - (1+\frac{1}{s})^k\sum_{i< t}C_{0,i} + (t-1)s\#(\boule{n}\setminus\boule{n-r})\]
  so 
  \[\sum_{0\leq i<t}C_{0,i} \leq \frac{TE_n^0}{(1+\frac{1}{s})^k} + \frac{(t-1)s\#(\boule{n}\setminus\boule{n-r})}{(1+\frac{1}{s})^k}.\]
  When $n$ grows, the first term of this upper bound converges to a constant independent of $t$ and of $x$ (from the definition of $TE_n^0$ because $E_n^0$ is at most the square of the size of $\boule{n}$), and the second term vanishes.
  This concludes the proof because it shows that ${\sum_{i=0}^\infty C_{0,i}}$ is bounded independently of $x$, and it counts the number of state changes in the trace ${\trac{F^2}{x}}$ and the trace ${\trac{F^2}{F(x)}}$.
\end{proof}

Let us show two counter-examples to Theorem~\ref{theo:majboundedchange} when hypothesis on either $\spac$ or the symmetry of neighborhood are removed:
\begin{enumerate}
\item let $\spac$ be the free group with $2$ generators, $V$ the symmetric neighborhood made of all generators and their inverse and element $\cspac$, and denote by $F$ the associated majority CA.
  Define the configuration $c^b$ for any ${b:\N\to\{0,1\}}$ by 
  \[c^b_g = b(\spaclen{g}).\]
  Because we are in the free group with $2$ generators, if ${g\in\boule{n}\setminus\boule{n-1}}$ then $g+V$ contains $3$ elements in ${\boule{n+1}\setminus\boule{n}}$ and $V$ is of size $5$.
  Therefore we have ${F(c^b)_g = b(\spaclen{g}+1)}$, and thus ${F(c^b)=c^{b'}}$ with ${b'(n) = b(n+1)}$ for all ${n\in\N}$.
  By choosing $b$ such that ${b(i+1)\cdots b(i+2k) = 0^k1^k}$ for arbitrarily large $k$ and for some $i$, we deduce that there is no ${p>0}$ such that ${F^p}$ is convergent.
\item for ${\spac=\Z^2}$, let us consider the non-symmetric neighborhood ${V=\{(0,0),(0,1),(1,0)\}}$ and let $F$ be the associated majority CA.
  Define the configuration $c^b$ for any ${b:\N\to\{0,1\}}$ by 
  \[c^b_{(x,y)} = b(x+y).\]
  Again, for any $(x,y)\in\Z^2$ there are two elements $z'$ in ${(x,y)+V}$ such that ${c^b_{z'}=b(x+y+1}$, so ${F(c^b)_{(x,y)}=b(x+y+1)}$.
  We can conclude as above that there is no ${p>0}$ such that $F^p$ is convergent.
\end{enumerate}

\subsection{Freezing CA}

The argument of Theorem~\ref{theo:majboundedchange} to prove the bounded-change property is not immediate and the property is generally undecidable, even when ${\spac=\Z}$ (Theorem~\ref{theo:bcundecidable}).

In some cases, however, the structure of the local rule of a CA directly implies the bounded-change property.
The so called \textit{freezing CA}, introduced in \cite{GolOlThey15}, are an example.

\begin{definition}\label{def:freezing}
  A CA ${F}$ on ${\confs}$ is \emph{freezing} if there is some partial order ${(Q,\leq)}$ on its state set such that for any ${c\in\confs}$ and any ${z\in\spac}$ it holds: 
  \[F(c)_z\leq c_z.\]
\end{definition}

Clearly, in a freezing CA the number of changes at any cell in any orbit is bounded by the depth of the partial order, so all freezing CA are bounded-change.
The freezing property can be tested efficiently by looking at the local transition table.
Indeed, given any CA $F$, we can compute the canonical state change graph ${(Q,\to)}$ defined by ${q_1\to q_2}$ if some transition changes state $q_1$ into $q_2$.
Then one can check that $F$ is freezing if and only if this graph is acyclic. 

Going back to our earlier metaphor, the partial order involved in the freezing property can be seen as a local non-increasing energy that serves as a certificate for the bounded change property.
Of course, not all bounded-change CA are freezing: for instance $F^2$, where $F$ is some CA from the hypothesis of Theorem~\ref{theo:majboundedchange}, is a bounded-change CA but it cannot be freezing since $F$ and hence $F^2$ commute with the transformation that permutes states $0$ and $1$.
Thus, if transition ${0\rightarrow 1}$ exists in the canonical state change graph of $F^2$, then transition ${1\rightarrow 0}$ also exists, which prevent any freezing order to satisfy Definition~\ref{def:freezing}.

There are many example of freezing CA studied in the literature \cite{GriMoo96,ulam,Fuentes,forestfire,Gravner98,BoSmUz15}.

A freezing CA $F$ for some order $\leq$ can also be \emph{monotone}: 
\[x\leq y \Rightarrow F(x)\leq F(y)\]
where $\leq$ here denotes the cellwise extension of $\leq$ to configurations.
Among these examples, the classical bootstrap percolation CA \cite{ChLeRe79} has been the starting point of a rich branch of percolation theory, which can be viewed as the study of the qualitative behavior of monotone freezing CA initialized on random Bernoulli configurations \cite{BoSmUz15,Holroyd03,BaBoPrSm16}.
This CA, denoted $F_B$ in the sequel, is defined on ${\{0,1\}^{\Z^2}}$ as follows: 
\[F_B(c)_z =
  \begin{cases}
    0&\text{ if $c_z=0$ or ${\#\bigl\{z'\in z+N:c_{z'}=0\bigr\}\geq 2}$,}\\
    1&\text{ else.}
  \end{cases}
\]

where ${N=\{(1,0),(0,1),(-1,0),(0,-1)\}}$.

$F_B$ is not nilpotent on periodic configurations (it admits both $\overline{0}$ and $\overline{1}$ as fixed points), and a classical result of \cite{Enter1987} shows that it is $\mu$-nilpotent for any full-support Bernoulli measure\footnote{\cite{Enter1987} doesn't use the notion of $\mu$-nilpotence but the property ${\mu(\{c:F^\omega(c)=\overline{0}\})=1}$. It can be shown that the two are equivalent for freezing CA in general \cite[Lemma 3.6]{salo2021bootstrap}.}.
Actually, much more is known on this CA started from random initial configurations \cite{Holroyd03}.
A qualitatively very different behavior is obtained by a similar CA, denoted $F_O$, and defined on ${\{0,1\}^{\Z^2}}$ as follows:
\[F_O(c)_z =
  \begin{cases}
    0&\text{ if $c_z=0$ or $c_{z+(1,0)}=0$ or $c_{z+(1,1)}=0$,}\\
    1&\text{ else.}
  \end{cases}
\]
$F_O$, contrary to $F_B$, is $\mu$-nilpotent for some Bernoulli measures but not all.

As a third example of monotone freezing CA with 2 states, let's consider $F_M$ defined as follows:
\[F_M(c)_z =
  \begin{cases}
    0&\text{ if $c_z=0$ or ${\#\bigl\{z'\in z+N:c_{z'}=0\bigr\}\geq 3}$,}\\
    1&\text{ else.}
  \end{cases}
\]

A classification of behaviors of all monotone freezing CA with $2$ states started from random Bernoulli configurations is given in \cite{BaBoPrSm16,BoSmUz15}.
The following theorem focuses on examples $F_B$ and $F_O$ to illustrate this classification.

\begin{figure}
  \centering
  \begin{minipage}{.24\linewidth}
    \begin{center}
      \includegraphics[width=.99\textwidth]{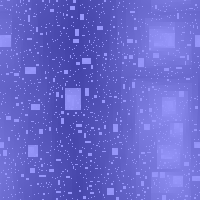}\\
      \small ${F_B}$ with ${\mu([0])=0.05}$
    \end{center}
  \end{minipage}
  \begin{minipage}{.24\linewidth}
    \begin{center}
      \includegraphics[width=.99\textwidth]{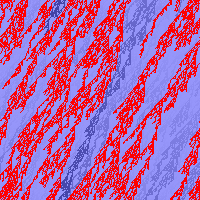}
      \small ${F_O}$ with ${\mu([0])=0.29}$
    \end{center}
  \end{minipage}
  \begin{minipage}{.24\linewidth}
    \begin{center}
      \includegraphics[width=.99\textwidth]{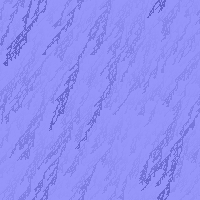}
      \small ${F_O}$ with ${\mu([0])=0.35}$
    \end{center}
  \end{minipage}
  \begin{minipage}{.24\linewidth}
    \begin{center}
      \includegraphics[width=.99\textwidth]{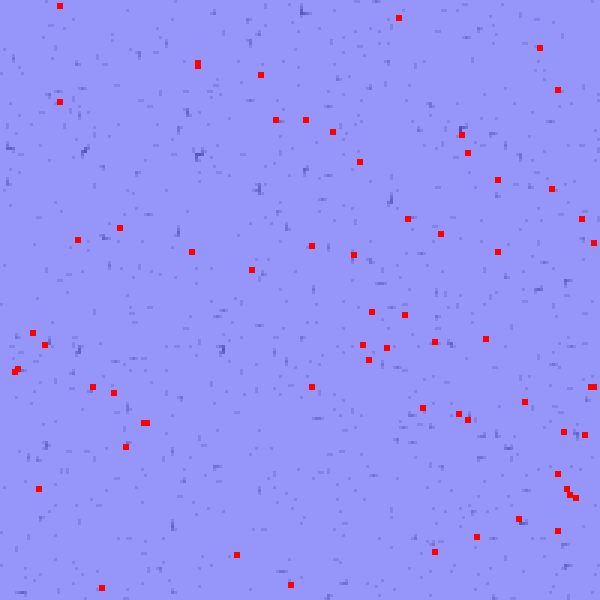}
      \small ${F_M}$ with ${\mu([0])=0.8}$
    \end{center}
  \end{minipage}
  \caption{Representation of fixed points reached under either $F_B$ or $F_O$ or $F_M$ starting from $\mu$-random initial configurations with Bernoulli measures $\mu$. The color interpretation is as follows: red represents state $1$ and shades of blue represent state $0$ with the convention that the lighter the color, the earlier the freezing time. This illustrates Theorem~\ref{theo:classicbootstrap}: $F_B$ converges almost surely to $\overline{0}$ no matter how small ${\mu([0])}$ is, while $F_O$ needs a large enough ${\mu([0])}$ to erase almost surely all $1$s, and $F_M$ almost never converges to $\overline{0}$, no matter how large ${\mu([0])}$ is.}
  \label{fig:bootstrap}
\end{figure}

\newcommand\bounda[1]{\mathcal{B}(#1)}

\begin{theorem}
  \label{theo:classicbootstrap}
  $F_B$ is generically nilpotent, and $\mu$-nilpotent for any translation-ergodic measure $\mu$ of full support.
  However, there is a threshold ${0<\theta<1}$ such that $F_O$ is $\mu$-nilpotent for Bernoulli measure $\mu$ when ${\mu([0])>\theta}$ and not $\mu$-nilpotent if ${\mu([0])<\theta}$. In particular $F_O$ is not generically nilpotent. Finally, $F_M$ is not $\mu$-nilpotent as soon as $\mu$ has full support.
\end{theorem}
\begin{proof}
  It is enough to show generic nilpotency of $F_B$ and the $\mu$-nilpotency follows by Theorem~\ref{theo:genmunil}.
  For any ${n\in\N}$ consider the set of configurations having an annulus of $0$ of radius $n$ around position ${(0,0)}$: 
  \[A_n = \{c\in\{0,1\}^{\Z^2} : \forall z, \|z\|=n\implies c_z=0\}.\]
  It is easy to show by induction that if ${\|z\|\leq n}$ and ${c\in A_n}$ then ${F_B^\omega(c)_z = 0}$ (the annulus of radius $n-1$ is first turned into $0$ starting from the corner, and so on).
  $A_n$ is open (it's a union of cylinders) and for any ${m\in \N}$ the union ${U_m = \bigcup_{n\geq m}A_n}$ is open and dense (any cylinder intersects $A_n$ for large enough $n$).
  Therefore the intersection ${I = \bigcap_{m\geq 0} U_m}$ is a comeager set.
  From the remark above, it holds ${F_B^\omega(c)_z=0}$ for any ${z\in\Z^2}$ and any ${c\in I}$.
  Therefore $I$ is included in the realm of attraction of ${\overline{0}}$ and we deduce that ${F_B}$ is generically nilpotent.

  Concerning $F_O$, the claim on Bernoulli measures for which it is $\mu$-nilpotent is enough to show that it is not generically nilpotent by Theorem~\ref{theo:genmunil}.
  One can show that ${F_O^\omega(c)_{(0,0)}= 1}$ if and only if there is a directed path ${(\pi_n)_{n\in\N}}$ in $\N^2$ starting from ${\pi_0=(0,0)}$ and such that ${\pi_{n+1}\in\pi_n+\{(0,1),(1,1)\}}$ for all ${n\in\N}$, and verifying ${c(\pi_n)=1}$ for all ${n\in\N}$.
  Indeed, if such path exists, it is easy to check that ${F^t(c)_{\pi_n}=1}$ for all ${t\geq 1}$ and all ${n\in\N}$ by induction, so in particular ${F^\omega_O(c)_{(0,0)}=1}$.
  On the contrary, if no such infinite path exists, then there is a finite bound $B$ on the longest finite directed path with the same properties (otherwise we could extract an infinite one by compacity).
  One can then check that ${F_O^{B+1}(c)_{(0,0)}=0}$.
  It is well-known from oriented percolation theory \cite{Durrett_1984} that when $c$ is a $\mu$-random configuration for Bernoulli $\mu$, there is a threshold ${0<\theta<1}$ such that an infinite path ${(\pi_n)_{n\in\N}}$ exists:
  \begin{itemize}
  \item with positive probability when ${\mu([0])<\theta}$,
  \item with probability $0$ when ${\mu([0])>\theta}$.
  \end{itemize}
  We deduce that ${F_O^t(\mu)([1])\to 0}$ so $F_O$ is $\mu$-nilpotent in the second case, while ${F_O^t(\mu)([1])\geq\epsilon>0}$ for all $t\geq 0$ so $F_O$ is not $\mu$-nilpotent in the first case.

  Finally, $F_M$ admits a finite obstacle: precisely, if a 2 by 2 square of $1$s appear in some configuration, it will never disappear.
  Thus if $\mu$ gives a probability larger than $\epsilon>0$ to a cylinder made of this pattern and also to the cylinder ${[0]}$, then so will ${F^t(\mu)}$ for any ${t\geq 1}$.
  We deduce that $F_M$ is not $\mu$-nilpotent.
  
\end{proof}

Although $\mu$-nilpotency is well understood for monotone freezing CA with $2$ states by bootstrap percolation theory \cite{morris_2017,Morris2017}, it remains hard for freezing CA in general.

\begin{theorem}[\cite{salo2021bootstrap}]
  The set of 2D freezing CA with $2$ states which are $\mu$-nilpotent for all full-support Bernoulli measure $\mu$ is recursively inseparable from the set of 2D freezing CA with $2$ states which are $\mu$-nilpotent for no full-support Bernoulli measure $\mu$.
\end{theorem}

A measure $\mu$ is said limit-computable if there is a computable map ${\phi}$ such that ${\displaystyle \lim_{n\to\infty}\phi(u,n)-\mu([u]) = 0}$ for each cylinder $[u]$.
One can check that if ${\displaystyle \mu' = \lim_{t\to\infty}F^t(\mu)}$ for some CA $F$ and $\mu$ a computable measure, then $\mu'$ is limit computable.
There is a remarkable converse to this observation.

\begin{theorem}[\cite{DelacourtM17,HELLOUIN_DE_MENIBUS_2016}]
  Let ${\spac=\Z^d}$ for some ${d\geq 1}$.
  Let $\mu$ be any translation-invariant and limit-computable measure on some alphabet.
  Then there exists a CA $F$ (on a possibly different alphabet) such that ${\displaystyle \mu=\lim_{t\to\infty} F^t(\mu_0)}$ where $\mu_0$ is the uniform Bernoulli measure on the alphabet of $F$.
\end{theorem}

Interestingly, the construction in the above theorem relies on CA that are not convergent, although the particular orbit starting from $\mu_0$ (and actually many other initial measures) is itself convergent.
We have no idea about whether it is possible to realize the same limit measures with convergent CA.

\begin{question}
  What are the limit measures ${F^\omega(\mu_0)}$ where $F$ is convergent (or bounded change, or freezing) and $\mu_0$ is the uniform Bernoulli measure (or some computable measure)?
\end{question}

\section{Computational Power of Convergent Cellular Automata}
\label{sec:compuconv}

The convergence property dramatically restricts the possible dynamics of CA as we have seen so far.
One can therefore legitimately ask whether the obvious fact that general CA are computationally universal still holds for convergent, bounded-change and freezing CA.
This section tackles this question and outlines an answer in three steps: yes, even freezing CA in dimension 1 are computationally universal, but there is some loss of complexity, and for a convergent CA $F$, it is interesting to study the map ${c\mapsto F^\omega(c)}$ from a computational point of view.

\subsection{Embeddings of computationally universal systems}

Let's start by the obvious observation that the space-time diagrams of any 1D CA can be grown by a 2D freezing CA.

\begin{example}
  \label{ex:freezing2Dsimu1D}
  Any 1D CA $F$ with states $Q$ and neighborhood $V$ can be simulated by a 2D freezing CA $F'$ with
  states ${Q\cup\{\ast\}}$ as follow. Let ${V'=\{(v,-1):v\in V\}}$. A cell in a state from $Q$ never
  changes. A cell in state $\ast$ looks at cells in its $V'$
  neighborhood: if they are all in a state from $Q$ then it updates to
  the state given by applying $F$ on them, otherwise it stays
  in $\ast$. Starting from a all-$\ast$ configuration except on one
  horizontal line where it is in a $Q$-configuration $c_0$, this 2D
  freezing CA will compute step by step the space-time diagram of $F$
  on configuration $c_0$.
\end{example}

In dimension 1, it is also possible to embed any Turing machine computation (with a polynomial slowdown) inside a convergent CA using the zigzag trick of Example~\ref{exa:zigzag}: by adding more states, one can make the head actually compute like a Turing head and do one transition at each zigzag (see \cite[Proposition 5]{ollinger:hal-02266916} for details).

For 1D freezing CA, the embedding of computationally universal systems is much more constrained but still possible using Minsky machine \cite{GolOlThey15,Carton2018,ollinger:hal-02266916}.

\newcommand\cll[4]{\draw[fill=#3] (#1,#2) rectangle ++(1,1) node[pos=.5] {#4};}
\newcommand\ceh[3]{\cll{#1}{#2}{green!30!white}{#3}}
\newcommand\cec[3]{\cll{#1}{#2}{blue!30!white}{#3}}
\newcommand\ces[3]{\cll{#1}{#2}{red!50!white}{#3}}
\newcommand\ceb[2]{\cll{#1}{#2}{white}{$b$}}
\newcommand\cew[2]{\cll{#1}{#2}{gray}{$w$}}
\newcommand\cei[3]{\cll{#1}{#2}{yellow!50!white}{#3}}
\newcommand\cedh[2]{\cll{#1}{#2}{white}{$\cdots$}}
\newcommand\cedv[2]{\cll{#1}{#2}{white}{$\vdots$}} 

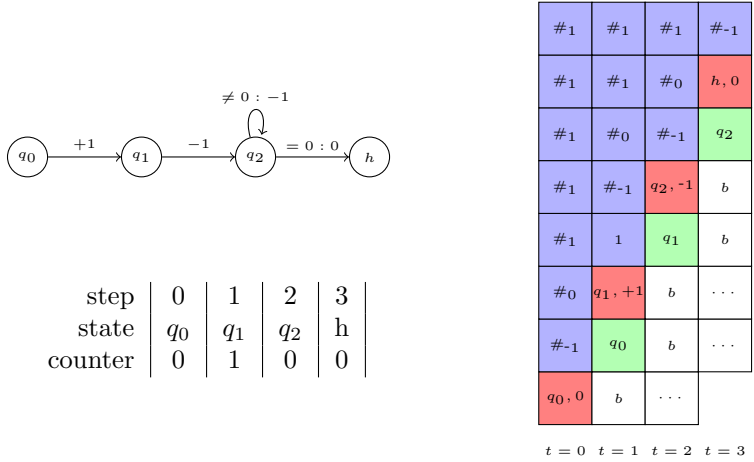
\begin{figure} 
  \begin{center}
    \begin{minipage}[c]{0.45\linewidth}
      \begin{center}
        \begin{tikzpicture}[scale=1]\tiny
          \tikzstyle{config} = [draw,circle,minimum size=15,fill=white]
          \node[config] (qz) at (0,0) {$q_0$};
          \node[config] (qu) at (1.5,0) {$q_1$};
          \node[config] (qd) at (3,0) {$q_2$};
          \node[config] (qh) at (4.5,0) {$h$};
          \draw[->] (qz) edge  node[above]{$+1$}  (qu);
          \draw[->] (qu) edge  node[above]{$-1$}  (qd);
          \draw[->] (qd) edge[loop above]  node[above]{$\neq 0 : -1$}  (qd);
          \draw[->] (qd) edge  node[above]{$=0 : 0$}  (qh);
        \end{tikzpicture}
        \end{center}\vskip 1cm
        \begin{center}
        \begin{tabular}{r|c|c|c|c|}
          step &0&1&2&3\\
          state & $q_0$&$q_1$&$q_2$ &h\\
          counter &0&1&0&0
        \end{tabular}
      \end{center}
    \end{minipage}
    \begin{minipage}[c]{0.5\linewidth}
      \begin{center}
      \begin{tikzpicture}[scale=.7]\tiny
        \def\-{\raisebox{.75pt}{-}}
        \draw (2.5,1.5) node {$t=0$};
        \draw (3.5,1.5) node {$t=1$};
        \draw (4.5,1.5) node {$t=2$};
        \draw (5.5,1.5) node {$t=3$};
        \ces{2}{2}{$q_0,0$}\ceb{3}{2}\cedh{4}{2} 
        \cec{2}{3}{$\#_{\-1}$}\ceh{3}{3}{$q_0$}\ceb{4}{3}\cedh{5}{3} 
        \cec{2}{4}{$\#_{0}$}\ces{3}{4}{$q_1,+1$}\ceb{4}{4}\cedh{5}{4} 
        \cec{2}{5}{$\#_{1}$}\cec{3}{5}{$1$}\ceh{4}{5}{$q_1$}\ceb{5}{5} 
        \cec{2}{6}{$\#_{1}$}\cec{3}{6}{$\#_{\-1}$}\ces{4}{6}{$q_2,\-1$}\ceb{5}{6} 
        \cec{2}{7}{$\#_{1}$}\cec{3}{7}{$\#_{0}$}\cec{4}{7}{$\#_{\-1}$}\ceh{5}{7}{$q_2$} 
        \cec{2}{8}{$\#_{1}$}\cec{3}{8}{$\#_{1}$}\cec{4}{8}{$\#_{0}$}\ces{5}{8}{$h,0$} 
        \cec{2}{9}{$\#_{1}$}\cec{3}{9}{$\#_{1}$}\cec{4}{9}{$\#_{1}$}\cec{5}{9}{$\#_{\-1}$} 
      \end{tikzpicture}
    \end{center}
    \end{minipage}
  \end{center}
  \caption{\label{fig:minsky}On the left, a Minsky machine with $1$ counter and a few steps of the execution starting from counter value $0$. On the right the corresponding space-time diagram of the freezing CA of Example~\ref{ex:minskyfreezing} encoding it. Counter value is encoded as the blue part in each trace: value $n$ is represented by state sequence ${1^n\#_{-1}\#_0\#_1}$. Green cells in the space-time diagram indicate position where all the required information to compute a Minsky transition is available locally. Cells containing the result of Minsky transitions are represented in red.}
\end{figure}

\begin{definition}
  A $k$-\emph{counter Minsky machine} is a $4$-tuple ${M=(Q_M,q_0,h,\tau)}$ where ${q_0,h\in Q_M}$ are the initial and halting states and \[\tau : Q_M\times\{0,1\}^k\rightarrow Q_M\times\{-1,0,1\}^k\] is its transition map, which verifies ${\tau(h,\cdot)=(h,(0,\ldots,0))}$. A \emph{configuration} of $M$ is an element of ${Q_M\times\N^k}$. $M$ transforms any configuration ${c=(q,(\chi_i)_{1\leq i\leq k})}$ in one step into configuration \[M(c)=(q',(\max(0,\chi_i+\delta_i))_{1\leq i\leq k})\] where ${(q',(\delta_i)_{1\leq i\leq k}) = \tau (q,(\min(1,\chi_i))_{1\leq i\leq k})}$. $M$ halts on input ${(\chi_i)_{1\leq i\leq k}\in\N^k}$ if there is a time $t$ such that ${M^t(q_0,(\chi_i)_{1\leq i\leq k})\in (h,\N^k)}$.
\end{definition}

\begin{example}\label{ex:minskyfreezing}
  Any $k$-counter Minsky machine $M$ can be embedded into a 1D freezing CA $F$ in the following sense. To simplify the exposition, let's take $k=1$.
  A configuration of $M$ at some step of the evolution is a pair ${(q,n)}$ where $q$ is a finite state and ${n\in\N}$ the current value of the counter.
  This configuration will be encoded into the trace of some cell ${z\in\Z}$ of $F$: the state $q$ will appear at cell $z$ at some time, and after the value $n$ will be encoded by the time interval between specific state changes occurring at $z$.
  Successive configurations of an orbit of $M$ are then encoded into successive cells of $F$ through the temporal traces.
  In order to allow $F$ to correctly produce an orbit where the trace at position $z+1$ encodes the configuration which is the image by $M$ of the configuration encoded into the trace at position $z$, the key point is to shift temporally the time interval containing the encoded configuration from position to position $z+1$.
  This technical condition allows to implement all the operations on a counter locally (zero test, increment and decrement).
  Concretely, if the configuration of $M$ at time $0$ is encoded into the time interval ${[0,\Delta_0]}$ at cell $0$, then the configuration at step $t$ of machine $M$ is encoded into the time interval ${[2t,2t+\Delta_t]}$ at cell $t$.
  See Figure~\ref{fig:minsky} for an example with details.
\end{example}

From the embedding of Minsky machines as in the above example, one expects undecidability results to be transferred to 1D freezing CA.
There are many ways to state such results, depending on the additional details implemented.
Let us mention the following one which essentially relies on a Minsky machine embedding but requires a more technical construction than the above example.
A freezing CA is always $k$-change for $k$ no smaller than the cardinal of its alphabet, but determining the minimal $k$ for which it is $k$-change is undecidable.

\begin{theorem}[{\cite[Theorem 6]{ollinger:hal-02266916}}]
  There exists a constant $k$ such that the following problem is undecidable: given a freezing 1D CA $F$, decide whether $F$ is $k$-change.
\end{theorem}

The intuition for the construction in the proof of the above theorem is as follows:
\begin{itemize}
\item first, the construction is such that the largest sequence of state changes that can possibly occur in some orbit at a given cell is unique, no other sequence of state changes can be as long as this one;
\item second, if this sequence appears in some orbit, it enforces the presence of a correct halting computation of the Minsky machine; conversely, if the simulated Minsky machine halts, there exists an orbit where this special sequence of state changes occurs at some cell;
\item third, the length of this sequence, denoted $L$, is independent of the number of states of the Minsky machine being simulated (intuitively, only the initial and halting states of the machine appear in it).
\end{itemize}

With these ingredients, one has a reduction from the halting problem of Minsky machines to the problem of whether a freezing CA is $(L-1)$-change.

\subsection{Dimension 1: complexity gap between bounded change and convergent CA}

There is a strong information flow bottleneck in bounded change CA that becomes critical in 1D.
It is not the case for convergent CA.
This can be formulated using communication complexity or classical computational complexity of the prediction problem \cite{vollmar81,GolOlThey15,ollinger:hal-02266916} (see also \cite{GolesMWT21} which generalizes the result to other $\spac$). The prediction problem associated to a CA asks for the state of a cell after a given time starting from a given input, a possible formalization is as follows: given an input pattern ${u\in Q^{\boule{tr}}}$ where $r$ is the radius of the considered CA, the problem is to determine the value of cell $\cspac$ after $t$ steps starting from a configuration ${c\in[u]}$ (note that by choice of the domain of $u$, this does not depend on $c$).
The prediction problem has communication complexity at most $\log(n)$ and is in class NL for 1D bounded change CA, while it can be $\Omega(\sqrt{n})$ and P-complete for 1D convergent CA.

Instead of presenting the above results and their formal setting, we are going to illustrate this difference between bounded change and convergent CA through the problem of recognizing palindromes.
Before stating the result, let's formalize the notion of language recognition by CA under time constraints (but without space constraint).
Say a language ${L\subseteq\{0,1\}^*}$ is recognized by a CA $F$ on ${Q^\Z}$ in time ${\tau:\N\to\N}$ if ${\{0,1,B,A,R\}\subseteq Q}$ (where $B$ is a blank state, $A$ an accepting state and $R$ a rejecting state) and for any ${n\in\N}$ and any ${u\in\{0,1\}^n}$, the orbit of configuration $c^u$ defined by 
\[c^u(z)=
  \begin{cases}
    u_z&\text{ if $0\leq z<n$},\\
    B&\text{ else,}
  \end{cases}
\]
is such that:
\begin{itemize}
\item there is ${t\leq\tau(n)}$ such that ${F^t(c^u)_0\in\{A,R\}}$,
\item for the minimal such time $t$ it holds 
  \[F^t(c^u)_0=A\iff u\in L.\]
\end{itemize}

Finally, PAL is the language of palindromes on alphabet ${\{0,1\}}$, \textit{i.e.} words $u$ such that ${u_i=u_{n-1-i}}$ for ${0\leq i< n=|u|}$.
The next theorem states that PAL is hard to recognize with bounded change CA, while it is easily recognized by convergent CA.
Note that limitations of bounded change (or bounded communication) CA were established previously in different settings: in \cite{vollmar81} language recognition is considered with a space limitation (only the cells initially containing the input can be modified), and in \cite{KutribM10a} the focus is on real-time recognition (\textit{i.e.} ${\tau(n)=n}$ ).

\begin{theorem}\label{theo:palrec}
  Let ${\spac=\Z}$. Suppose that a $k$-change CA can recognize PAL in time ${\tau(n)}$, then the recognition time is at least exponential, \textit{i.e.} ${\tau(n)\geq \alpha^n}$ for some ${\alpha>1}$ and large enough $n$. However, there exists a convergent CA that recognizes PAL in quadratic time.
\end{theorem}
\begin{proof}
  Suppose by contradiction that a $k$-change CA $F$ on ${Q^\Z}$ and of radius $r$ recognizes PAL in time ${\tau(n)\leq\alpha^n}$ with ${\alpha<2^{\frac{1}{4rk}}}$.
  For ${n\in\N}$ and ${u\in\{0,1\}^n}$ denote by ${T_z^{u,n}}$ the prefix of the trace of length ${\tau(n)}$ and of width ${2r}$ starting at position $z$ in the orbit of $c^u$: 
  \[T_z^{u,n}: t\in\{0,\ldots,\tau(n)\}\mapsto F^t(c^u)_z,\ldots, F^t(c^u)_{z+2r}.\]
  By the $k$-change property, there are at most ${(|Q|\tau(n))^{2rk}\leq |Q|^{2rk}\alpha^{2rkn}}$ such prefixes (for all choices of $u$ and $z$) because for each of the $2r$ columns of the trace, there are at most $k$ changes and each change can be described by a new state and a time step between $1$ and $\tau(n)$.
  On the other hand, there are ${2^{\frac{n}{2}}}$ words in PAL of even length $n$, so for large enough even $n$ and by choice of $\alpha$ there are ${u\neq v}$ in PAL such that 
  \begin{equation}
    T_{n/2}^{u,n}= T_{n/2}^{v,n}\label{eq:eqtrace}
  \end{equation}
  Then, consider the word ${w = u_0\cdots u_{n/2-1}v_{n/2}\cdots v_{n-1}}$. By the Equality~(\ref{eq:eqtrace}) above, we have that ${\trac{F}{c^u}(t)=\trac{F}{c^w}(t)}$ for all ${0\leq t\leq\tau(n)}$. So in particular $F$ must accept $w$ in time at most $\tau(n)$ because it does for $u$. However, $w$ is not a palindrome, which yields the desired contradiction.

  It is not difficult to adapt Example~\ref{exa:zigzag} to make a convergent CA that recognizes PAL in quadratic time.
  Roughly, using extra states, any input word ${BuB}$ is turned in one step into a segment of state with two components: the binary component that keeps the information of $u$ and the zigzag component of the form ${\tikz[baseline,scale=.3]{\stateZ{0}{0}\stateOZO{1}{0}\stateOO{2}{0}}\cdots\tikz[baseline,scale=.3]{\stateOO{0}{0}\stateZ{1}{0}}}$, where the head state $\tikz[baseline,scale=.3]{\stateOZO{0}{0}}$ holds an extra bit, initialized to $u_0$.
  Then, with the zigzag movement of the head, this bit held inside the right-moving head $\tikz[baseline,scale=.3]{\stateOZO{0}{0}}$ is compared to the bit in the binary component when the head bounces on the right border (\textit{i.e.} $u_0$ is compared to $u_{n-1}$ at the end of the first zig).
  If the comparison fails, an error state $e$ appears which is interpreted as a rejecting state.
  If not, the head turns into $\tikz[baseline,scale=.3]{\stateZZO{0}{0}}$ (which does not hold any additional information) and goes back to the left boundary as in the rule of Example~\ref{exa:zigzag}.
  Once the left boundary is reached again, the head becomes $\tikz[baseline,scale=.3]{\stateOZO{0}{0}}$ and copies the bit from the binary component at this position (\textit{i.e.} $u_1$ at the start of the second zig).
  When the zone is shrunk down to a single position, an accepting state is generated that spreads to the left to reach the position which initially was holding value $u_0$ (the accepting state spreads over any other state except $e$, which is not present in the orbit of $c^u$ if $u$ is in PAL).
  The additional mechanism on top of the rule of Example~\ref{exa:zigzag} does not compromise the convergent property: the analysis is the same as in Example~\ref{exa:zigzag} with the presence of two spreading states instead of one (rejecting state $e$ spreading over the accepting state which spreads over any other state).
\end{proof}

Note that the above theorem does not say anything about the existence of a bounded-change CA recognizing PAL in super-exponential time.
It would be surprising that such a CA exists because the information to compare between the two halves of a palindrome is too dense to be dealt with and transmitted within the bounded change constraint, however we have no proof of this impossibility.

\begin{question}
  Is there a bounded change CA recognizing PAL (with arbitrary time and space) or any language with linear communication complexity in the sense of \cite{hromkovic97}?
\end{question}

We don't think that the quadratic time recognition of PAL by convergent CA is optimal in Theorem~\ref{theo:palrec}, however we have no idea on how to recognize PAL in real-time (\textit{i.e.} ${\tau(n)=n}$) with a convergent CA. More generally, we don't know whether the recognition time of languages by convergent CA can always be as good as the recognition time for general CA.

\subsection{Complexity of limit configurations}

Any convergent CA ${F}$ on $\confs$ induces a map ${F^\omega:\confs\to\confs}$.
We say that a configuration ${c\in\confs}$ is computable if it is computable as a map.
For any $t$, the map ${F^t}$ transforms computable configurations into computable configurations, however there is no reason to expect that the same holds for $F^\omega$.

Let us first make the following observation: the computability of the limit configuration is linked to the computability of freezing times.

\begin{lemma}
  \label{lem:freezetime}
  Let $F$ be a convergent CA over $\confs$ and ${c\in\confs}$ be any computable configuration. Then the map ${z\mapsto F^\omega(c)_z}$ is Turing reducible to the map ${z\mapsto\freezt{F}{c}{z}}$.
  Moreover, if $F$ is freezing, both maps are Turing-equivalent.
\end{lemma}
\begin{proof}
  To compute ${F^\omega(c)_z}$ from ${\freezt{F}{c}{z}}$, it is sufficient to compute ${F^{\freezt{F}{c}{z}}(c)_z}$.
  Reciprocally, if $F$ is supposed freezing, $\freezt{F}{c}{z}$ can be computed from ${F^\omega(c)_z}$ as it is the first time step $t$ such that ${F^t(c)_z=F^\omega(c)_z}$.
\end{proof}

For bounded change 1D CA, it is impossible to produce uncomputable limits from computable initial configurations.
The argument of the following proof is due to G. Richard.

\begin{theorem}[\cite{ollinger:hal-02266916}]
  \label{thm:computablelimits}
  For any 1D bounded-change CA $F$ and any computable configuration $c$, ${F^\omega(c)}$ is computable.
\end{theorem}
\begin{proof}[sketch]
  We restrict without loss of generality to $F$ of radius $1$ and fix a computable configuration $c$.
  If we know the number of state changes in the orbit of $c$ at cells $z_1$ and $z_2$, with ${z_1<z_2}$, then we can compute the value ${F^\omega(c)_z}$ for all ${z\in[z_1,z_2]}$.
  Indeed, we can compute ${F^t(c)}$ on cells ${[z_1,z_2]}$ for increasing values of $t$ until the correct number of changes is observed at $z_1$ and $z_2$.
  From that time on, the evolution of cells in the segment ${[z_1,z_2]}$ is independent of the context since cells $z_1$ and $z_2$ no longer change and $F$ has radius $1$. So we can compute the evolution until this set of cells reaches a fixed point, which is then the value they have in ${F^\omega(c)}$.

  The number of changes at any cell is bounded, and there is a maximal value $L$ for which there are infinitely many positions ${z<0}$ with exactly $L$ changes. Moreover, there is a limit position $z_L$ to the left of which no cell has more than $L$ changes. The same is true for positions ${z>0}$ giving the corresponding constants $R$ and $z_R$.
  
  Then the algorithm to compute ${F^\omega(c)_z}$ given $z$ is the following: compute larger and larger portions of the space-time diagram around position $z$ until finding ${z_1\leq z \leq z_2}$:
  \begin{enumerate}
  \item ${z_1\leq z_L}$ and ${z_R\leq z_2}$,
  \item the state of $z_1$ has changed $L$ times and the state of $z_2$ has changed $R$ times.
  \end{enumerate}
  Then it is sufficient to apply the algorithm of the above to compute ${F^\omega(c)_{[z_1,z_2]}}$ and therefore obtain ${F^\omega(c)_z}$.
\end{proof}

Interestingly, the above limitation for bounded change CA disappears when considering convergent CA.

\begin{theorem}[\cite{ollinger:hal-02266916}]
  There exists a 1D convergent CA $F$ and a computable configuration $c$ such that ${F^\omega(c)}$ is uncomputable.
\end{theorem}

\begin{figure}
  \centering
  \includegraphics[width=.25\textwidth]{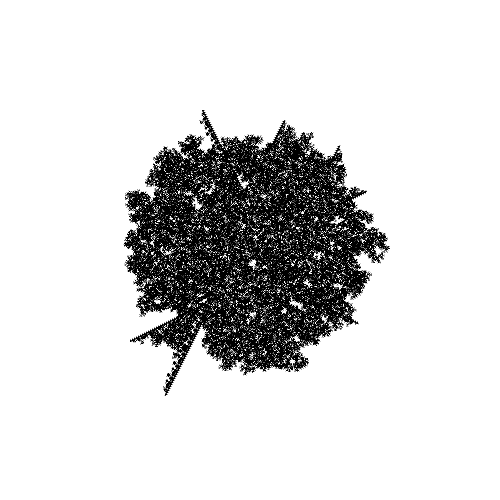}
  \includegraphics[width=.25\textwidth]{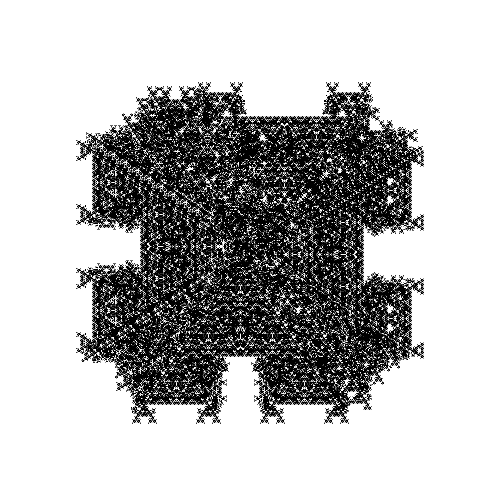}
  \includegraphics[width=.25\textwidth]{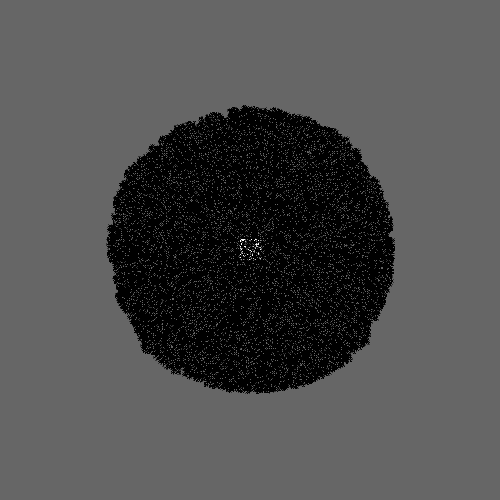}
  \vskip 1mm
  \includegraphics[width=.25\textwidth]{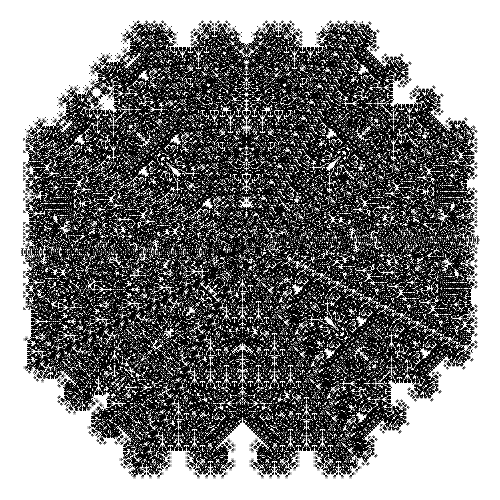}
  \includegraphics[width=.25\textwidth]{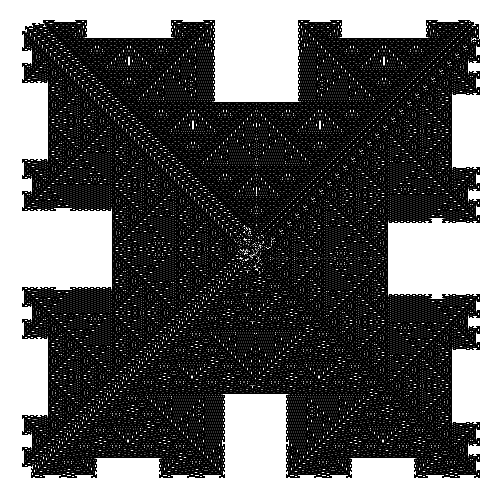}
  \includegraphics[width=.25\textwidth]{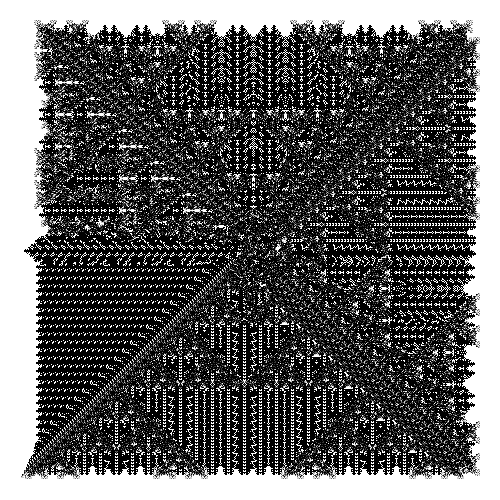}
  \caption{Examples of configurations obtained after some time from a finite seed by randomly chosen 2D freezing CA. All CA have the same neighborhood and the same number of states, with the same freezing order. All finite seed have the same size (the non-uniform portion of the configuration is a small centered square), and all example configurations were obtained after the same number of steps. The difference in the uniform background comes from the fact that not all examples have the same quiescent states.}
  \label{fig:randomfreezing}
\end{figure}

In 2D, it is easy to find a freezing CA $F$ and a computable configuration $c$ such that ${F^\omega(c)}$ is uncomputable: let $c(i,j)$ be $1$ if ${i,j\in\N}$ and Turing machine $i$ halts in at most $j$ steps on the empty tape, and $0$ otherwise.
Then, if $F$ is the freezing CA that spreads state $1$ vertically, we get that ${F^\omega(c)_{(i,0)}=1}$ for ${i\in\N}$ if and only if machine $i$ halts.

A more interesting question is to ask whether ${F^\omega(c)}$ can be uncomputable for some finite or eventually periodic configuration $c$ (\textit{i.e.} a configuration which is totally periodic up to a finite region of ${\Z^2}$).
The behavior of freezing CA from finite configuration is quite rich (see Figure~\ref{fig:randomfreezing}) and it is actually possible to realize uncomputable limits with a freezing $1$-change CA and slightly more with a $2$-change CA.

\begin{theorem}\cite{jCCSA,ollinger:hal-02266916}\label{theo:limitconfhard}
  For a configuration ${c\in Q^{\Z^2}}$ and ${q\in Q}$, define the set of cells of $c$ in state $q$: ${\chi_q(c)=\{z\in\Z^2 : c(z)=q\}}$. The following holds:
  \begin{itemize}
  \item There exists a $1$-change freezing CA $F_1$, a state $q$ of $F_1$, and a finite
    configuration $c^1$ such that ${\chi_q(F_1^\omega(c^1))}$ is not
    computable. 
  \item For any $1$-change freezing CA $F$, any state $q$ and any finite configuration $c$, ${\chi_q(F^\omega(c))}$ is either recursively enumerable or co-recursively enumerable.
  \item There exists a $2$-change freezing CA $F_2$, a state $q$ of $F_2$, and a
    finite configuration $c^2$ such that ${\chi_q(F_2^\omega(c^2))}$
    is neither recursively enumerable, nor co-recursively enumerable.
  \end{itemize}
\end{theorem}
\begin{proof}
  The first item is the main result of \cite{jCCSA}.
  The second item is \cite[Proposition 9]{ollinger:hal-02266916}.
  The third item is \cite[Corollary 1]{ollinger:hal-02266916}.
\end{proof}

Recall from Lemma~\ref{lem:freezetime} that in order to produce an uncomputable limit configuration from a computable one by a freezing CA, the freezing times of cells must be uncomputable.
We stress that for this reason it is not sufficient for a 2D freezing CA to be \emph{computationally universal} in order to produce uncomputable limit configurations from eventually periodic initial ones.
For instance, in the orbits described in the simple embedding of Example~\ref{ex:freezing2Dsimu1D}, the freezing times are computable by construction since cells are frozen progressively line by line.
We don't have a general understanding of the ingredients behind the capacity of freezing CA to produce uncomputable limit configuration starting from simple (computable) initial ones.

\begin{figure}
  \centering
  \includegraphics[width=.7\textwidth]{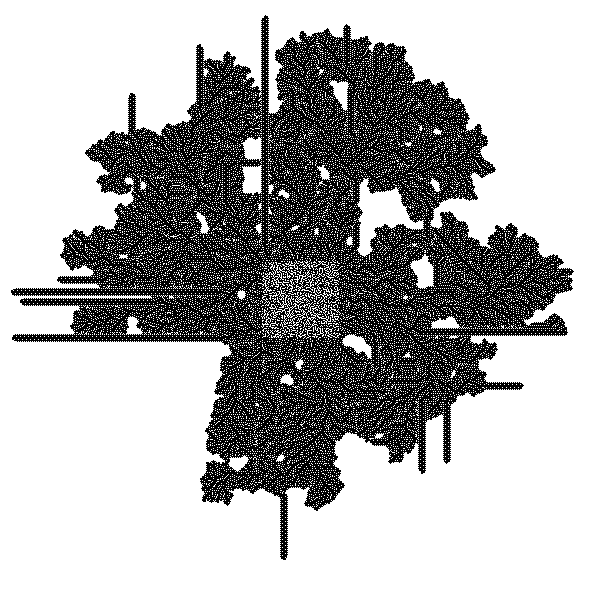}
  \caption{Configuration reached from a random finite initial configuration (center) by the CA game of life without death after a few hundreds of steps. Black represents state $0$, white represents state $1$.}
  \label{fig:lwd}
\end{figure}
Concretely, one can consider the \emph{game of life without death} CA \cite{GriMoo96}, which is a CA on the space on configurations ${\{0,1\}^{\Z^2}}$ defined by\footnote{For the reader familiar with life-like CA, note that we use the somewhat unconventional representation of alive cells by $0$ to stay consistent with Definition~\ref{def:freezing}.}: 
\[F(c)_z =
  \begin{cases}
    0\text{ if $c_z=0$ or $Z_8(c,z)=3$},\\
    1\text{ else,}
  \end{cases}
\]
where ${Z_8(c,z)}$ counts the number of occurrences of state $0$ among the 8 neighbors of cell $z$ in configuration $c$. 
The prediction problem of this CA was shown to be P-complete in \cite{GriMoo96}, however its ability to produce complex limit configurations from simple initial ones has not been established to our knowledge (see Figure~\ref{fig:lwd}).

\begin{question}
  Let $F$ be the \emph{game of life without death} CA. Is there a finite or eventually periodic configuration $c$ such that ${F^\omega(c)}$ is uncomputable? More generally, for which freezing CA is this possible? What are sufficient or necessary conditions to achieve this behavior?
\end{question}

\section{Some related works and further reading}
\label{sec:related}

What was presented so far is of course not an exhaustive view of cold dynamics in cellular automata, and it also does not tackle other dynamical systems.
Let us finish this tutorial by pointing out some related research directions that does not exactly fit in the formalism or classes seen so far.

\paragraph{Randomization and measure rigidity in CA.} There have been much work on CA considered as dynamical systems acting on probability measures \cite{Pivato09}.
Let us mention two (linked) aspects that extends what we presented above: recall that unique ergodicity (Definition~\ref{def:uniqergo}) is both the property of having a unique invariant measure, and the property that any initial measure converges in Cesaro mean towards this unique invariant measure under the CA.
If one considers a surjective CA on ${\spac=\Z^d}$, then the uniform Bernoulli measure is an invariant measure \cite{Bartholdi2010}.
It can never be the only one (a CA always admits some linear combination of dirac measures as an invariant measure), but a natural question is to ask under which additional conditions it becomes unique, a phenomenon called \emph{measure rigidity} (see \cite[section 2]{Pivato09} and references therein).
Note that the famous Furstenberg conjecture can be formulated as measure rigidity question on a particular CA \cite[section 2F]{Pivato09}.
Another much studied phenomenon is called \emph{randomization} and studied since the 80s \cite{miyamoto79,Lind84}: it is the property that a large class of initial measures (including Bernoulli measures of full support) converge in Cesaro mean towards the uniform Bernoulli measure (see \cite[section 3A]{Pivato09} and references therein).
More recently, a characterization of randomization was obtained for Abelian CA and an example was shown where the convergence is not only in Cesaro mean but actually a direct convergence (in the topology on measures) \cite{HELLOUIN_DE_MENIBUS_2018}.
Recall that this difference in the kind of convergence plays also a role in the links between $\mu$-nilpotency and unique ergodicity as discussed in Section~\ref{sec:nil}.
These two phenomena are now rather well understood on CA with an algebraic structure, but the generalization to other kind of CA remains largely open.

\paragraph{Probabilistic CA, noise and ergodicity.} A probabilistic CA is a CA whose local evolution rule is probabilistic, thus randomness can be introduced at each time step (contrary to deterministic CA acting on probability measures where randomness only appears at the initial time step) \cite{Mairesse_2014}.
In this context, the property of unique ergodicity has been much studied, in particular by asking if the unique invariant measure is also attractive which is then called \emph{ergodicity} \cite{Chassaing_2011} and which can be seen as an analog of asymptotic nilpotency.
Among the many examples of probabilistic CA, a particularly interesting family are those obtained from a deterministic one by adding a (small) noise at each transition.
In that setting, ergodicity intuitively means that the noise will progressively erase any information from the initial configuration, and non-ergodicity is a form of resistance to noise.
It appears that in many cases the resulting noisy CA is ergodic \cite{ergonoisCA} as a probabilistic CA.
In fact, proving that a particular noisy CA is \emph{not} ergodic turns out to be difficult.
A famous example in dimension 2 is a non-symmetric majority rule which was proven to be non-ergodic in \cite{toom1980stable}.
The result in \cite{toom1980stable} is in fact more general: any monotone deterministic CA which is nilpotent over finite configurations yields a non-ergodic probabilistic CA when adding a small enough noise.
Thus in dimension 2, very simple rules can produce CA resistant to noise.
In dimension 1, the only such example known is extremely complicated \cite{G_cs_2001} (recall that the construction technique involved in this result inspired the construction of Theorem~\ref{theo:nilpernotue}).

\paragraph{Self-assembly tilings.} In the field of DNA computing, several models have been proposed that describe how a spatially extended structure (an \emph{assembly}) grows according to some local rules.
For instance, the aTAM model \cite{Winfree_phd}, usually considered on the grid $\Z^2$, uses the formalism of Wang tiles with glues to define the local process of growth: an assembly is a (usually finite) set of occupied positions of $\Z^2$ and the aTAM rules determine what tiles can be added on an empty position in the neighborhood of the assembly.
A key parameter of the model is \emph{temperature} which determines the amount of glues necessary to allow a tile to be added.
We skip the details of the definition, but they are inspired by biochemistry and one of the strengths of the model is that it is both theoretically interesting and realistic enough to allow in vitro experiments.
We refer to \cite{Patitz_2013} for a global overview (including many other self-assembly models).
Any aTAM system can be seen as a $1$-change (asynchronous and non-deterministic) process acting on ${Q^{\Z^2}}$ where $Q$ is the set of tiles plus a special symbol representing an empty position.
Note that in the case of a deterministic aTAM system (called \emph{directed}), one can even define a freezing (1-change) CA $F$ that produces the same limit configurations (see proof of Theorem~\ref{theo:limitconfhard}).
Important progress have been accomplished recently in understanding the computational power of aTAM depending on the temperature parameter \cite{DBLP:conf/stoc/MeunierRW20,DBLP:conf/stoc/MeunierW17}.


\paragraph{Automata networks.} We worked with $\spac$ infinite, but in some sense we already considered the case of a finite space of cells when introducing nilpotency on periodic configurations (Definition~\ref{def:nilper}): indeed, the orbit of a totally periodic configuration on $\Z^d$ can be seen as an orbit of configurations on a finite space of cell which is a torus (because since CA commute with translations, periods are preserved under iterations).
One can go much further by considering 
automata networks which are non-uniform CA on a space of cells which can be an arbitrary graph \cite{MP43,Goles90}.
Actually, Theorem~\ref{theo:majboundedchange} has its roots in the automata network settings \cite{PaperGoles}.
In this context, a particularly fruitful approach is to analyze how the graph of cells influences the dynamics of the network (see \cite{Gadouleau_2019} and references therein for an overview).
Following this approach, many results focused on nilpotency or convergence towards a unique attractor \cite{Rob86,Ara08,Gadouleau_2016}. 
More recently, freezing automata networks where also considered \cite{GolesMWT21,GolesMMO18,Goles_2013}.

\section*{Acknowledgment}
We thank the anonymous referees for their careful reading, their many suggestions and corrections, and in particular, for having corrected us on the fact that the bootstrap percolation CA $F_B$ is generically nilpotent (Theorem~\ref{theo:classicbootstrap}).

\bibliographystyle{plainurl}
\bibliography{paper}

\end{document}